\documentclass[11pt]{amsart}
\usepackage{enumerate,fancyhdr,longtable,geometry,amsmath,amsxtra,amssymb,latexsym,amscd,amsthm}
\usepackage{bbold}
\usepackage{xcolor}
\usepackage{tikz}
\newtheorem{theorem}{Theorem}[section]
\newtheorem{proposition}[theorem]{Proposition}
\newtheorem{lemma}[theorem]{Lemma}
\newtheorem{corollary}[theorem]{Corollary}
\theoremstyle{definition}
\newtheorem{definition}[theorem]{Definition}

\newtheorem{remark}[theorem]{Remark}

\newtheorem{remarks}[theorem]{Remarks}
\newtheorem{application}[theorem]{Application}

\numberwithin{equation}{section}

\newcommand{\eee}{{\rm e}}
\newcommand{\trees}{\mathcal{T}}

\newcommand{\card}[1]{\left| #1 \right|}

\usepackage{anysize}
\usepackage{fancyhdr}

\DeclareMathOperator{\argmax}{argmax}
\newcommand{\absnu}{\card\nu\!}

\newcommand{\boldmu}{\boldsymbol{\mu}}
\newcommand{\boldxi}{\boldsymbol{\xi}}
\newcommand{\boldPsi}{\boldsymbol{\Psi}}
\newcommand{\boldPhi}{\boldsymbol{\Phi}}
\newcommand{\boldGamma}{\boldsymbol{\Gamma}}
\newcommand{\boldBeta}{\boldsymbol{\mathcal B}}
\newcommand{\boldCeta}{\boldsymbol{\mathcal C}}

\newcommand{\boldFeta}{\boldsymbol{\mathcal F}}
\newcommand{\Beta}{\mathcal B}

\newcommand{\boldzero}{\boldsymbol{0}}
\newcommand{\boldone}{\boldsymbol{1}}
\newcommand{\tg}{g}
\newcommand{\wg}{\widehat g}

\newcommand{\ext}[1]{\left[#1\right]}
\newcommand{\rser}{R[\![X]\!]}
\newcommand{\rlaur}{R(\!(\!X\!)\!)}
\newcommand{\kser}{K[\![X]\!]}
\newcommand{\klaur}{K(\!(\!X\!)\!)}
\newcommand{\res}{\mathrm Res}
\newcommand{\sdag}[1]{\underline{#1}^\dagger}
\newcommand{\utilde}[1]{\underline{\tilde{#1}}}
\newcommand{\uhat}[1]{\underline{\widehat{#1}}}

\begin{document}
\title{CONVERGENCE OF CLUSTER AND VIRIAL EXPANSIONS\\ FOR REPULSIVE CLASSICAL GASES}
\author{Roberto Fern\'andez}
\address{New York University Shanghai}
\curraddr{Room 1149-3, 1555 Century Avenue, 
Pudong New Area, Shanghai, 
China 200122}
\email{rf87@nyu.edu}
\thanks{}
\author{Nguyen Tong Xuan}
\address{Gran Sasso Science Institute}
\curraddr{Viale F. Crispi, 7 67100 L'Aquila}
\email{tongxuan.nguyen@gssi.it}
\thanks{}

\date{\today}

\dedicatory{}
\keywords{ cluster and virial expansions, repulsive gases}

\begin{abstract} 
We study the convergence of cluster and virial expansions for systems of particles subject to positive two-body interactions. Our results strengthen and generalize existing lower bounds on the radii of convergence and on the value of the pressure.  Our treatment of the cluster coefficients is based on expressing the truncated weights in terms of trees and partition schemes, and generalize to soft repulsions previous approaches for models with hard exclusions.  Our main theorem holds in a very general framework that does not require translation invariance and is applicable to models in general measure spaces.  For the virial coefficients we resort to an approach due to Ramawadth and Tate that uses Lagrange inversion techniques only at the level of formal power series and leads to diagrammatic expressions in terms of trees, rather than the doubly connected diagrams traditionally used.  We obtain a new criterion that strengthens, for repulsive interactions, the best criterion previously available (proposed by Groenveld and proven by Ramawadth and Tate).  We illustrate our results with a few applications showing noticeable improvements in the lower bound of convergence radii.
 
\end{abstract}
\maketitle

\section{Introduction}
The \emph{virial expansion} is the expansion of the pressure in powers of the density.  It was introduced by Kamerlingh Onnes at the beginning of the XX-th century, to replace the van der Waals equation of state for non-ideal gases.  The expansion both solves physical inconsistencies of the Van der Waals approach and provides a more faithful description of diluted non-ideal fluids.  While the coefficients of the expansion are useful already at the phenomenological level (see tables, e.g., in \cite{tables_virial,wiki}), statistical mechanics can be used to determine them starting from microscopic models.  This is traditionally accomplished through an auxiliary expansion of the pressure in powers of an effective parameter called \emph{fugacity} ---the \emph{cluster expansion} introduced by Mayer in the early forties.  Differentiation of this series yields, in turns, the density as a power series of the fugacity.  The virial expansion is, in principle, obtained, by formally inverting this last series to obtain the fugacity as a power series in the density, and composing the latter with the expansion for the pressure.

The expressions for the coefficients of the cluster and virial expansion are well known for systems with two-body interactions.  They involve some non-trivial combinatorics whose bookkeeping is more efficiently done in terms of diagrams.  The terms of the cluster expansions can be expressed in terms of trees \cite{pen67,Fern}, while those of the virial expansion are usually written in terms of doubly connected diagrams (see, for instance, the classical reference \cite{uhfor63}).  The complexity, and number of possibilities for virial diagrams increase very rapidly with the order of the term, and already the computations for the simplest non-ideal gas ---hard spheres--- are a subject of current interest \cite{SJA,RH,CNMcoy}.

From the point of view of rigorous statistical mechanics, however, the crucial issue is the determination of the region in parameter space for which the cluster and virial series converge.  Such region correspond to regimes in which fluids remain in gas form, without exhibiting transitions to liquid or solid phases.  The study of convergence properties of both series started in the late sixties \cite{pen67,GK,LPen64} but did not gather momentum till a couple of decades later.  The convergence of the cluster expansion has been studied by a variety of methods: Kirkwood-Salzburg equations \cite{GK} \cite[Chapter 4]{Rue69}; tree-graph bounds \cite{bry84}, induction methods \cite{kotpre86,dob96,dob96a} and partition schemes \cite{Fern} (see \cite{bisferpro10} for further comparisons of these methods).  The last method yields the best bounds to date and have been subject of a number of applications and improvements (e.g.\ \cite{jps,temmel14}).  

The case of gases in continuum space is of particular interest here.  The convergence of their cluster expansion depends on the interplay of two effects:
\begin{itemize}
\item[(E1)] The alternation of signs in the terms defining the coefficients of the cluster expansion.
\item[(E2)] The tug-of-war between the repulsive and attractive parts of the interaction. 
\end{itemize}
The second effect is particularly subtle, and has not been properly exploited till very recently \cite{proyuh17}.  The first effect is the only one present for repulsive potentials and it is the one taken into account in the transcriptions of the inductive \cite{poguel09} and Kirkwood-Salzburg \cite{Sabine2018} approaches.  The approach based on partition schemes has been successfully applied to the extreme case of hard repulsions \cite{fps} and, also, in the seminal work \cite{proyuh17}.  In the latter, however, the repulsive part of the potential is bounded below by zero and for the purely repulsive case the resulting bound is no better than the classical one.  
\medskip 

Mayer's approach turns the proof of convergence of virial expansions into a combination of cluster-expansion convergence with estimations of coefficients obtained by inversion of a series.  After the ground breaking work of Lebowitz and Penrose \cite{LPen64} this inversion is invariably done by resorting to the Lagrange inversion formula for analytic series.  This strategy, however, applies only to regimes in which the cluster expansion converges and, hence, can potentially introduce unphysical constraints inherited from the unphysical singularities that limit the convergence of the cluster expansion.  In fact, virial convergence beyond cluster-expansion convergence has, indeed, been observed in some particular models \cite{Jen15}.  In this work we present expressions for the virial coefficients (Proposition \ref {prop:rr.10}) obtained on a purely combinatorial basis, without relying on cluster-expansion convergence.  Consequently, our main estimates of the radius of convergence of the virial expansion  (Theorem \ref{th:rr1}) can potentially exceed our best estimate for the cluster radius of convergence (see Remark \ref{rem:radius}). 
 
The objective of this paper is to develop the best available results for cluster and virial convergence of systems with positive two-body interactions.  The paper has, therefore, two largely independent parts which, however, share some common technology.  Cluster expansion results rely in expressing the truncated weights in terms of trees.  This is a time-honored tradition that, however, can be made more effective, as pointed in \cite{Fern}, through the use of partition schemes.  As we discuss below, these schemes allow the incorporation of further dilution effects due to the repulsive character of the interaction.  This leads to a very general convergence result (Theorem \ref{th:cluster.gen} below) which generalizes previous results for systems with exclusions \cite{Fern,fps}. The theorem also strengthen, in the purely repulsive setting the best results available for soft interactions \cite{uel04,poguel09,proyuh17}.  [We have not investigated its relation with the technically more involved approach in \cite{temmel14}.]
 For the benefit of the reader having physical applications in mind we start by presenting out results in more usual but less general forms (Theorems \ref{th:cluster} and \ref{th:cluster.1}).

Our treatment of the virial expansion is a strengthening of the unfortunately yet unpublished work by   S.\ Ramawadth and S.\ Tate \cite{ram15,ramtat15}.  We follow very close their novel approach, adding improvements specific for the repulsive case.   A noteworthy aspect of this approach is that it leads to alternative diagrammatic expressions of the virial coefficients [formulas  \eqref{eq:bell.43} and  \eqref{eq:mer30} below], based on collections of trees rather than the more involved doubly connected graphs.  The approach relies on two ingredients.  The first ingredient is a careful handling of the different coefficients through the formalism of formal power series and formal Laurent series.  Below, we present a detailed account of the different step involved, including a largely self-contained review of the relevant aspects of the theory of formal series.  The second ingredient introduced in \cite{ram15,ramtat15} is the very clever procedure of \emph{merging of trees}, also fully discussed below.  The combination of both ingredients leads to the improved lower bounds on the radii of convergence of the virial expansion detailed in Theorem \ref{th:rr1} below.  Unfortunately, the improvement applies only to repulsive interactions.  We are at present unable to extend our treatment to interactions to both repulsive and attractive parts.  

Ramawadth's and Tate's results also provides a proof of an early convergence criterion proposed by Groeneveld~\cite{gro} in his 1967 PhD dissertation.  This criterion, which was presented without proof, is an improvement of Lebowitz' and Penrose's for positive interactions. To our knowledge, Ramawadth's and Tate's constitute its  first published proof.  As we explain below, Theorem \ref{th:rr1} constitutes a strengthening of Groenveld's in the case of positive interactions.  The full criterion, however, is difficult to apply in practice.  We therefore propose a slightly weaker criterion (Corollary \ref{cor:rr100}) that is, however, computationally more efficient.  This ``efficient" criterion is not directly comparable with Groenveld's but, as we show in several examples, it is often much better.
\medskip

We illustrate our criteria with a number of examples and applications (Sections \ref{sec:cluster.app} and \ref{ssec:virial.eff}).  For models of hard spheres our estimates of cluster radii of convergence for $d=1$ are 60\% better than the classical bound ---and within 20\% of the exact value--- and 40\% better for $d=2$.  This last result was already obtained in \cite{fps}.  Our estimates for the virial radii for hard spheres, on the other hand, are 40\%  and 25\% better than that of Greonveld (itself 60\% larger than the classical estimate by Lebowitz and Penrose \cite{LPen64}) respectively for $d=1$ and $d=2$. For a power-law repulsive potential of interest in physics~\cite{Barlow12} our bound of the virial radius is 10\% better than Groeneveld's and 80\% better than the classical estimates.  

\section{Assumptions and definitions}

Our starting point is the grand canonical partition function for a classical gas on a finite region $\Lambda\subset \mathbb{R}^d$:

\begin{equation}\label{eq:virial1}
\Xi_{\Lambda}(z)=1+\sum_{N=1}^{\infty}\frac{z^N}{N!}\,\int_{\Lambda^N}\eee^{-\beta U_N(x_1,\ldots,x_N)}d\,x_1\ldots d\,x_N\;.
\end{equation}
Here $\beta$ is the inverse temperature, $z$ is a parameter called \emph{fugacity} and the functions $U_N$ define the interaction potential of the particles in the system.  The results in the sequel remain true ---module measurability conditions--- with the Lebesgue measure $dx$ replaced by any other $\sigma$-finite measure.  In particular, the measure could include space dependent fugacities. 
In this work we assume that the interactions are: 
\begin{itemize}

 \item \emph{Two-body:}
 \[U_N(x_1,\ldots,x_N)\;=\;\sum_{1\le i<j\le N}\phi(x_i,x_j)\;,\]
 with measurable functions $\phi$ which are allowed to take the value $+\infty$ to describe \emph{hard-core exclusions}.  
  
\item  \emph{Repulsive:} 
 \[ \phi(x_i,x_j)\;\ge\; 0\;. \]
 
\item \emph{Translation invariant:}
  \[ \phi(x_i,x_j)\;=\; \phi\bigl(\left|x_i-x_j\right|\bigr)\;. \]
In particular the functions $U_N$ are invariant under permutations of their arguments. 

\item \emph{$C(\beta)$-Tempered:} The function 
\begin{equation}\label{eq:rr.cbeta}
C(\beta)\;:=\;\int_{\mathbb{R}^d}\left|\eee^{-\beta\phi(x,0)}-1\right| dx
\end{equation}
is finite.  
 \end{itemize} 
 
 The function $C(\beta)$ play on important r\^ole in the estimates of convergence below.  As remarked in  \cite{Rue69}, its finiteness for some $\beta$ is equivalent to $\phi(x,0)$ being integrable outside some set of finite Lebesgue measure (for instance, outside the set $\{x:\phi(x)\le 1\}$). This, in turns, implies that finiteness of $C(\beta)$ for some $\beta$ implies finiteness for all $\beta$.  

The finite-volume pressure is
\[\beta P_{\Lambda}(z)\;=\;\frac{1}{|\Lambda|}\ln \Xi_{\Lambda}(z).\]
while the thermodynamic pressure  $p$ is obtained through the thermodynamic limit 
\begin{equation}\label{eq:virial2}
p\;=\;\lim\limits_{\Lambda\uparrow\mathbb{R}^d}P_{\Lambda}(z)\;
\end{equation}
e.g.\ in Fisher sense.   The finiteness of $C(\beta)$ is, technically, not enough to guarantee the existence of this limit.  A convenient sufficient condition is, for instance \cite[Chapter 3]{Rue69}
the existence of $A, r_0>0$ and $\alpha>d$ such that $\phi(0,x)\le A/|x|^\alpha$ for $|x|>r_0$.  This condition ---called \emph{temperedness} by Ruelle--- is satisfied by all physical examples, including the ones discussed below.

The \emph{finite-volume cluster expansion} is the Taylor expansion of $\beta P_{\Lambda}$ around $z=0$ (diluted regime):
\begin{equation}\label{eq:rr.fce}
\beta P_{\Lambda}(z)\;=\;\sum_{n=1}^{\infty}\frac{b_n(\Lambda)}{n!}z^n\;.
\end{equation}
The \emph{finite-volume virial expansion} is obtained by expanding, instead, 
\begin{equation}\label{eq:virial3}
\beta P_{\Lambda}(\rho_\Lambda)\;=\;\sum_{n=1}^{\infty}\frac{\beta_n(\Lambda)}{n!}\rho^n_\Lambda\;.
\end{equation}
where $\rho_\Lambda$ is the  mean density of particles for the grand-canonical ensemble in the volume $\Lambda$
\begin{equation}\label{eq:rr.fce1}
 \rho_{\Lambda}(z)\;=\; z\frac{\partial [\beta P_{\Lambda}]}{\partial z}\;, 
\end{equation}
so to yield a correct version of the flawed van der Waals equation of state.  The virial expansion amounts to composing the cluster expansion \eqref{eq:rr.fce} with the function $z(\rho_\Lambda)$ obtaining by inverting \eqref{eq:rr.fce1}.  For $z$ within the radius of convergence of the cluster expansion, this inversion can be performed through Lagrange inversion formula.  The dependence of the virial coefficients $\beta_n$ as a function of the cluster coefficients $b_n$ can, instead, be obtained in a purely algebraic fashion, without analytical considerations, by working at the level of formal power series (see e.g.\ \cite{wiki_formal}).   Indeed, $z(\rho_\Lambda)$ can be represented by the formal power series obtained as the inverse series to  
the term-by-term derivative of \eqref{eq:rr.fce}, namely
\begin{equation}\label{eq:rr.fce2}
 \rho_{\Lambda}\;=\;\sum_{n=0}^{\infty}\frac{b_{n+1}(\Lambda)}{n!}z^{n+1}\;.
\end{equation}
The virial expansion follows from the composition of the cluster expansion with the inverse series of \eqref{eq:rr.fce2}. 

The eventual objects of interest is the thermodynamic limit \eqref{eq:virial2} and, thus, the infinite-volume limits of the coefficients
 the \emph{cluster expansion:} Denoting
\begin{eqnarray}\label{eq:rr.cec}
b_n&=& \lim\limits_{\Lambda\uparrow\mathbb{R}^d}b_n(\Lambda) \\
\label{eq:rr.cev}
\beta_n&=& \lim\limits_{\Lambda\uparrow\mathbb{R}^d}\beta_n(\Lambda)\;,
\end{eqnarray}
the corresponding (infinite-volume) \emph{cluster} and \emph{virial expansion} take the form
\begin{eqnarray}\label{eq:rr.ce}
\beta p &=& \sum_{n=1}^{\infty}\frac{b_n}{n!}z^n \\
\label{eq:rr.ve}
\beta p &=& \sum_{n=1}^{\infty}\frac{\beta_n}{n!}\rho^n 
\end{eqnarray}
These expansions, and the previous ones should be understood as formal power series and the equalities are conditioned to the actual convergence of the series.  For the sake of completeness we review formal power series and their operations in Section \ref{sec:pf.virial}.

The goal of the present work is to determine fugacity and density radii below which all the series of 
\eqref{eq:rr.fce}, \eqref{eq:virial3}, \eqref{eq:rr.ce} and \eqref{eq:rr.ve} converge uniformly, uniformly for all $\Lambda$.  

\section{Results on cluster expansions}

\subsection{Basic convergence results}

The following theorem summarizes the convergence criterion for the cluster expansion of systems of objects subject to a positive two-body interaction.
Let 
\begin{equation}\label{eq:rn5.1}
\varphi_n(x;x_1,\ldots,x_n)\;=\; 
\prod_{1\le i\le n} \bigl[1-\eee^{-\beta\phi(x,x_i)}\bigr]
\prod_{1\le i<j\le n} \eee^{-\beta\phi(x_i,x_j)}\;,
\end{equation}
\begin{equation}\label{eq:rr10}
\tg(n)\;:=\; \int_{ \mathbb{R}^{dn}} \varphi_n(0;x_1,\ldots, x_n)\,d\,x_1\cdots d\,x_n \;.
\end{equation}
and
\begin{equation}\label{eq:rn8.0}
\Psi(\mu)\;=\; 1 + \sum_{n\ge 1} \frac{\mu^n}{n!} \, \tg(n) 
\end{equation}
(\emph{vertex sum}).

 \begin{theorem}\label{th:cluster}

Assume $z\in\mathbb C$ such that 
\begin{equation}\label{eq:rr9.0}
\card z\;\le\; \frac{\mu}{\Psi({\mu})}
\end{equation}
for some $\mu>0$.  Then

\begin{itemize}
\item[(i)] The finite-volume cluster expansions \eqref{eq:rr.fce} and their limit \eqref{eq:rr.ce} converge uniformly.
\item[(ii)] Furthermore,   
\begin{equation}\label{eq:rr9.0.01}
 \card{\beta P_\Lambda(z)},\,  \card{\beta p(z)} \le\mu\;.
\end{equation}
\end{itemize}

\end{theorem}
This theorem is a particular case of the more general Theorem \ref{th:cluster.gen} below, which will be proven in Section \ref{ssec:rrproof.gen}.   A slightly weaker form of the theorem [sharp inequalities in \eqref{eq:rr9.0} and \eqref{eq:rr9.0.01}] admits a much simpler proof that we present, for completeness, in Section \ref{ssec:rrproof.eas}.

\begin{remark}\label{rem:rr10}
The $C(\beta)$-temperedness implies that $1\le \Psi(\mu)\le \eee^{\mu C(\beta)}$.  In the trivial case in which all interactions are zero, $C(\beta)=0$ and $\Psi(\mu)=1$.  This implies $\beta P_\Lambda(z)=\beta p(z) =z$ which is a series with infinite radius of convergence.  Otherwise, $\Psi$ grows at least quadratically with $\mu$ and, hence, the ratio $\mu/\Psi(\mu)$ is a continuous and positive function on $[0,\infty)$, taking the value zero for $\mu=0$ and converging to zero as $\mu\to\infty$.  As a consequence:

\begin{itemize}
\item[(a)] The ratio $\mu/\Psi(\mu)$ achieves its maximum in $[0,\infty)$.

\item[(b)] If $\card z$ is less or equal than this max, the fixed point equation $\card z=\mu/\Psi(\mu)$ has at least one solution. 
\end{itemize}
\end{remark}
These observations lead to the following corollary.

\begin{corollary}\label{cor:rr1} \  

\begin{itemize}
\item[(i)] The radii of convergence of the cluster expansions \eqref{eq:rr.fce} and \eqref{eq:rr.ce}  are bounded below by
\begin{equation}\label{eq:rn9.1.bis.bis}
\mathcal{R}_{\rm cluster}\;\ge\;  r^* \;:=\; \max_{\mu\ge 0} \frac{\mu}{ \Psi({\mu})}\;,
\end{equation}

\item[(ii)] We have the uniform bounds 
\begin{equation}\label{eq:rn9.1.2.bis}
 \card{\beta P_\Lambda(z)},\,  \card{\beta p(z)} \;\le \;\mu^*:=\argmax_{\mu\ge 0} \frac{\mu}{\Psi({\mu})} \quad \mbox{for all} \quad \card z \le r^*\;.
\end{equation} 

\item[(iii)] If $\card z\le r^*$, then  
\begin{equation}\label{eq:rn9.0.0.bis}
\card{\beta P_\Lambda(z)}, \card{\beta p(z)} \;\le\; \mu_z\quad \
\end{equation}
where $\mu_z$ the smallest solution to the fixed-point equation
\begin{equation}\label{eq:rn9.0.0.bis.bis}
\card z \Psi(\mu_z)\;=\; \mu_z\;.
\end{equation}
\end{itemize}
\end{corollary}

\subsection{More detailed cluster convergence bounds}\label{ssec:rr.10}

For the purposes of this paper it is important to highlight some refinements of the above results that are, usually, consider technical details hidden in the proof of the convergence criteria.   These refinements play an important role in the analysis of virial expansions presented below, and help to understand approximate schemes yielding criteria that are weaker but computationally more feasible than the convergence condition \eqref{eq:rr9.0}.  We start by pointing out the well known fact that the cluster expansion takes the form
\begin{equation}\label{eq:rr31.0}
\ln \Xi_\Lambda(z)\;=\; \sum_{n\ge 1} \frac{z^n}{n!} \int_{\Lambda^n} \omega^T_n(x_1,\ldots, x_n)\,dx_1\cdots dx_n
\end{equation}
Let $\mathcal G(n)$ denote the complete graph with vertices $1,\cdots,n$.  Then 
\begin{equation}\label{eq:rr32}
\omega^T_n(x_1,\ldots, x_n)\;=\; 
\left\{\begin{array}{cl}
\displaystyle 1 & n=1\\[5pt]
\displaystyle
\sum_{G\subset_{cc} \mathcal G(n)} \prod_{\{i,j\}\in E(G)} [w(x_i,x_j)-1] & n\ge 2
\end{array}\right.
\end{equation}
with $w(x_i,x_j)=-\eee^{\beta \phi(x_i,x_j)}$ and
``$\subset_{cc}$" standing for ``connected spanning subgraph of" and $E(G)$ being the set of edges of $G$.  The refinements of interest here involve sums over trees.  

Let us denote $\mathcal T^0[n]$ the set of trees with vertices $\{0,1,\ldots,n\}$ rooted in 0.  Each tree $\tau\in \mathcal T^0[n]$ defines the corresponding \emph{tree distance} which, in turns, gives rise to the obvious definition of ancestors and descendants of a vertex and, in particular, the notion of parent (= immediately preceding ancestor) and child (= immediate descendant) of a vertex.  For a vertex $i$ we shall denote $s_i$ the number of its children (= \emph{siblings}), and label those  by a double index: $(i,1),\ldots,(i,s_i)$.  Denote
\begin{equation}\label{eq:rn5.gen.1}
B_n(x_0;x_1,\ldots,x_n)\;=\; \sum_{\tau\in\mathcal T^0[n]} \prod_{i=0}^{n}
\varphi_{s_i(\tau)}(x_i;x_{(i,1)},\ldots,x_{(i,s_i(\tau))})
\end{equation}
and
\begin{equation}\label{eq:rr41.1.gg}
\Beta( z)\;:=\; 1+\sum_{n\ge 1} \frac{{ z}^n}{n!} \int_{S^n} B_{n}(0;x_1,\ldots, x_n)\,d\,x_1\cdots d\,x_n
\end{equation}
(\emph{tree sum}).

The following is the strengthening of Theorem \ref{th:cluster} which is actually proven below in an even more general version.
 \begin{theorem}\label{th:cluster.1} 

Assume $z\in\mathbb C$ satisfies condition \eqref{eq:rr9.0}, namely
$\card z\le \mu/\Psi(\mu)$ for some $\mu>0$.  Then, the following inequalities hold:   
\begin{equation}\label{eq:rr300.ex}
 \card{\beta P_\Lambda(z)},\,  \card{\beta p(z)} \;\le\; \card z\,\Beta(\card z)\;\le\;\mu\;.
\end{equation}

\end{theorem}

As a matter of fact, the proof reveals a sequence of improvements of the rightmost bound.  While these improvements have not been exploited yet, let us state them here for the sake of completeness.
 \begin{theorem}\label{th:cluster.2} 
 Consider $z\in\mathbb{C}$ satisfying  \eqref{eq:rr9.0}.  Define the function $\Pi_z:[0,\infty)\to[0,\infty)$
 \begin{equation}\label{eq:rr301}
 \Pi_z(x)\;=\; \card z\,\Psi(x)
\end{equation}
and let $\Pi_z^n$ denote its $n$-fold composition of $\Pi_z$ with itself. Then,
 \begin{itemize}
 \item[(i)] The sequence $(\Pi_z^n)_{n\ge 1}$ is strictly decreasing and converges to a limit $\Pi_z^\infty$ which is a fix point of $\Pi_z$.
 
  \item[(ii)] The following sequence of bounds hold
  \begin{equation}\label{eq:rr302}
  \card z\,\Beta(\card z)\;\le\;\Pi_z^\infty(\mu) \;\le\; \cdots \;\le\;\Pi_z^n(\mu)
  \;\le\;\cdots\;\le\; \Pi_z(\mu)\;\le\; \mu\;.
\end{equation}
 \end{itemize}
 \end{theorem}

Both Theorems \ref{th:cluster.1} and \ref{th:cluster.2} follow immediately from the more general Theorem \ref{th:cluster.gen} below.  Taken together they 
explain the usual approach of arriving to explicit bounds by weakening the interaction through the remotion of some of its terms.  In doing so, one obtains a larger tree sum $\widetilde\Beta\ge \Beta$ whose convergence provides a sufficient criterion.

\begin{corollary}\label{cor:cluster.1}
Let $\phi$ and $\widetilde\phi$ be two-body positive functions not identically zero such that 
 \begin{equation}\label{eq:rr303}
 0\;\le\; \widetilde\phi(x,y) \;\le\; \phi(x,y) \qquad,\quad x,y\in S\;.
 \end{equation}
Denoting with tildes objects corresponding to $\widetilde\phi$, we have 
\begin{itemize}
 \item[(i)] The radii of convergence of the cluster expansions \eqref{eq:rr.fce} and \eqref{eq:rr.ce}  are bounded below by
\begin{equation}\label{eq:rn9.1.bis}
\mathcal{R}_{\rm cluster}\;\ge\;  \widetilde r^* \;:=\; \max_{\mu\ge 0} \frac{\mu}{ \widetilde\Psi({\mu})}\;,
\end{equation}

 \item[(ii)] For $\card z\le \widetilde r^*$ we have
 \begin{eqnarray}\label{eq:rr300}
 \card{\beta P_\Lambda(z)},\,  \card{\beta p(z)} &\le& \card z\,\Beta(\card z)
 \;\le\; \card z\,\widetilde\Beta(\card z)\nonumber\\
 &\le& \widetilde\Pi_z^\infty(\mu)
 \;\le\; \cdots \;\le\;\widetilde\Pi_z^n(\mu)
  \;\le\;\cdots\;\widetilde\Pi^2_z(\mu)\;\le\; \widetilde\Pi_z(\mu)\;\le\; \mu\;.
\end{eqnarray}
\end{itemize}
\end{corollary}

\begin{remark}
A perhaps less known fact is the existence of rigorous \emph{upper} bounds of the radius of convergence of the cluster expansion.  Penrose \cite{pen63} obtained the sequences of bonds (valid more generally for stable two-body potentials)
\begin{equation}\label{eq:rrp1.pen}
\mathcal{R}_{\mathrm{cluster}}\;\le\; \Bigl[\frac{n}{(n-1)\,(b_n/n!)}\Bigr]^{1/(n-1)}
\end{equation}
(with analogous expressions for finite $\Lambda$).  This sequence of bound is sharp in that it has a subsequence converging to the actual radius  convergence.  The bound for $n=2$ can be strengthened \cite{pen63,gro62}
\begin{equation}\label{eq:rrp2}
\mathcal{R}_{\mathrm{cluster}}\;\le\; \frac{1}{\card{b_2}}\;=\; \frac{1}{C(\beta)}\;.
\end{equation}

\end{remark}

\subsection{Applications of cluster-expansion results}\label{sec:cluster.app}

Let us exploit Corollary \ref{cor:cluster.1}.

\begin{application}[{\bf The classical bound}]\label{app:rr.1}
The simplest application is obtained through the bound
\begin{equation}\label{eq:rr10.1} 
\tg(n) \;\le\; \bigl[C(\beta)\bigr]^n\;,
\end{equation}
obtained by bounding above $\eee^{-\beta\phi(x_i,x_j)}\le 1$ in \eqref{eq:rn5.1}.  In this case \eqref{eq:rn9.1.bis} implies 
 \begin{equation}\label{eq:rn.12}
\mathcal{R}_{\rm cluster}\;\ge\; r_1\;:=\;\max_{\mu\ge 0} \mu\,\eee^{-\mu C(\beta)}\;=\; \frac{1}{\eee\, C(\beta)}\;=\;\frac{0.367879}{C(\beta)}\;.
\end{equation}
This is the classic bound, given for instance in \cite[Chapter 4]{Rue69}, which has remained the only one available for positive interaction, except for the case of hard spheres \cite{fps} (in particular, it coincides with the bound given in \cite{proyuh17} for interactions with no negative part).  Part (ii) of the Corollary \ref{cor:cluster.1} yields the bounds
\begin{equation}\label{eq:rr21} 
\card{\beta P_\Lambda(z)},\card{ \beta p(z)} \;\le\;  -\frac{1}{C(\beta)}\,W\bigl(-C(\beta)\card z\bigr)\;\le\;
\frac{1}{C(\beta)}\qquad,\quad \mbox{for } \card z\le \frac{1}{\eee\, C(\beta)}\;.
\end{equation}
Here $W$ is Lambert's function (inverse of $x\eee^x$). 
\end{application}

\begin{application}[{\bf First correction to the classical bound}]\label{app:rr.2}
A further level of approximation is obtained by bounding by 1 some, rather than all, factors $\eee^{-\beta\phi(x_i,x_j)}$ in \eqref{eq:rn5.1}.  For instance, if this is done for all factors with $1\le i\le n$ and $n+1\le j\le n+m$ we obtain the inequality
\begin{equation}\label{eq:rr13}
\tg(n+m)\;\le\; \tg(n)\,\tg(m)\;.
\end{equation}
In particular,
\begin{eqnarray}\label{eq:rr14}
\tg(n)\;\le\; \left\{\begin{array}{ll}
[\tg(2)]^{n/2} & \mbox{if $n$ is even}\\[5pt] 
C(\beta)\,[\tg(2)]^{(n-1)/2} & \mbox{if $n$ is odd}
\end{array}\right.
\end{eqnarray}
These bounds lead, by \eqref{eq:rn9.1.bis}, to the condition
\begin{equation}\label{eq:rr15}
\mathcal{R}_{\rm cluster}\;\ge\; \max_{\mu\ge 0} \frac{\mu}{\cosh\bigl(\mu \sqrt {\tg(2)}\bigr)
+ \frac{C(\beta)}{\sqrt {\tg(2)}}\,\sinh\bigl(\mu \sqrt {\tg(2)}\bigr)}\;.
\end{equation}
To get an idea of how this compares with the classical bound \eqref{eq:rn.12}, we simply choose $\mu=1/\sqrt {\tg(2)}$.  This yields
\begin{equation}\label{eq:rr16}
\mathcal{R}_{\rm cluster}\;\ge\; r_2\;:=\;
\frac{1}{ \sqrt {\tg(2)}\,\eee +\bigl[C(\beta) - \sqrt {\tg(2)}\bigr]\sinh(1)}\;.
\end{equation}
The improvement with respect to the previous application is made explicit by the relation
\begin{equation}\label{eq:rr17}
\frac{r_2}{r_1}\;=\;\frac{1}{\delta+\frac{(1-\delta)}{2}(1-\eee^{-2})}\quad,\quad \delta=\frac{\sqrt {\tg(2)}}{C(\beta)}\;.
\end{equation}
For the extreme case of hard discs in $\mathbb{R}^2$ we have $\delta=3^{3/4}/(2\sqrt\pi)$ \cite{fps} and
$r_2/r_1\sim 1.25$, while the full optimization of the ratio $\mu/\Psi(\mu)$ leads to a convergence radius $1.39 \,r_1$ \cite{fps}.
By (ii) of Corollary \ref{cor:cluster.1} we also have the bounds $\beta P_\Lambda(z), \beta p(z) \le1/\sqrt { g(2)}=1/[\delta C(\beta)]$.  These are bounds larger than those in \eqref{eq:rr21}, but they apply uniformly in the larger region $\card z\le r_2$.

\end{application}

\begin{application}[{\bf Hard spheres}]\label{app.hard}
Models with pure hard-core interactions are the examples with maximal repulsion and, thus, those that should exhibit more clearly the improvements of our formula.  In the case of hard spheres, the interactions are of the form
\begin{equation}
\phi(x,y):=\left\{\begin{matrix}
+\infty &\, \text{if}\card{x-y}\le a \\ 
 0& \, \text{if}\card{x-y}>a
\end{matrix}\right.,
\end{equation}
for some hard-core radius $a>0$.  For these models $C(\beta)=V_d(a)$, the volume of the $d$-dimensional sphere of radius $a$, and
the coefficients $\tg(n)$ of the function $\Psi$ are of the form 
\begin{equation}\label{eq:rr240}
\tg(n)\;:=\; [V_d(a)]^n\,\wg_d(n)
\end{equation}
where $\wg_d(n)$ is a purely geometrical factor that depends on the dimension $d$ but not on the radius $a$:
\begin{equation}\label{eq:rr241}
\wg_d(n)\;:=\; \frac{1}{[V_d(1)]^{n}}\,\int_{\card{y_i}\le 1 \atop \card{y_i-y_j}>1} dy_1\cdots dy_n\;.
\end{equation}
As a consequence, the bound \eqref{eq:rn9.1.bis} becomes
\begin{equation}\label{eq:rr221.1.1hard}
\mathcal{R}_{\mathrm{cluster}}\;\ge\; R_d^{\mathrm{hard}}\;:=\;
 \frac{1}{V_d(a)}\,\max_{\alpha\ge 0}\,
\frac{\alpha }{\widehat\Psi_d(\alpha)}
\;.
\end{equation}
where the function
\begin{equation}\label{eq:rr221.2.1hard}
 \widehat\Psi_d(\mu) \;=\; 1+\sum_{n\ge 1} \frac{\mu^n}{n!}\,\wg_d(n)
 \end{equation}
 is in fact a polynomial with degree equal to the maximum number of non-overlapping unit spheres that can be made to overlap a fixed sphere.

The case $d=1$ corresponds to hard intervals in a line and one easily finds 
\begin{equation}\label{eq:rr221.3.1hard}
\widehat\Psi_1(\alpha)=1+ \alpha+\frac18 \alpha^2\;,
 \end{equation}
 hence $\alpha/\widehat\Psi_d(\alpha)$ achieve its maximum at $\alpha=\sqrt 8$ and 
 \begin{equation}\label{eq:rr221.4.1hard}
 R_1^{\mathrm{hard}}\;=\; \frac{1}{V_1(a)}\,\frac{\sqrt 2}{1+\sqrt 2}\;=\;\frac{0.5858}{V_1(a)}\;,
 \end{equation}
 60\% bigger than the classical bound \eqref{eq:rn.12}.  In fact, the model is solvable and the exact radius of convergence of the cluster expansion is $0.736/V_1(a)$ \cite{pen63}.
 
 The case $d=2$ was worked out in \cite{fps}; formula \eqref{eq:rr221.1.hard} leads to $R_2^{\mathrm{hard}} \gtrapprox 0.5107/V_2(a)$; almost 40\% above the classical bound.
 \end{application}
 
\begin{application}[{\bf Power-law potentials}]\label{app.power}
There has been recent interest \cite{richard,Barlow12, tai2011} in the phase diagram of repulsive potentials of the form
\begin{equation}\label{eq:rr1.power}
\phi_n(x,y)=\varepsilon\left(\frac{\sigma}{\card{x-y}}\right)^n
\end{equation}
To test condition \eqref{eq:rn9.1.bis} we considered the case $\beta=\epsilon=\sigma=1$, $n=6$ and $d=3$. To determine the function $\Psi$ we computed up to four decimal places
\begin{equation}\label{eq:rr2.power}
 g(2)= 0.6917C(1)\ ,\  g(3)= 0.3685C(1)\ ,\  g(4)= 0.145C(1)\ ,\   g(5)= 0.0627C(1)\;,
\end{equation}
with 
\begin{equation}\label{eq:rr3.power}
g(1)\;=\;C(1)\;=\;4\pi\int\limits_{r>0}\left[1-\eee^{-\left(\frac{1}{r}\right)^n}\right]r^2dr,
\end{equation}
and bounded $g(6)$ to $g(14)$ combining the values in \eqref{eq:rr2.power} with inequality \eqref{eq:rr13}.  We obtain, numerically, $R^*\gtrapprox 0.428/C(1)$.

 \end{application}

\subsection{Convergence of more general cluster expansions}\label{ssec:cluster.general}

Cluster expansions can be generalized in two directions: (1) Considering measure spaces other than $(\mathbb{R}^d,dx)$, and (2) considering complex-valued interactions.  Our convergence proof, in fact, is done in this doubly generalized framework.   More precisely, we consider a measure space $(S,\nu)$ ---with $\nu$ possibly complex valued---and complex-valued measurable functions $w(\cdot,\cdot): S^2\to \{u\in\mathbb{C}: \card u\le 1\}$ invariant under the exchange of arguments.  The corresponding generalized grand-canonical partition functions are of the form
\begin{equation}\label{eq:rr30}
\Xi_\Lambda(\nu)\;=\; 1+ \sum_{n\ge 1} \frac{1}{n!} \int_{\Lambda^n} \omega_n(x_1,\ldots,x_n)\,\nu(dx_1)\cdots \nu(dx_n)
\end{equation}
for each $\Lambda\subset S$ of finite $\card\nu$ measure, with
\begin{equation}\label{eq:rr30.1}
\omega_n(x_1,\ldots,x_n)\;=\;\prod_{1\le i <j\le n} w(x_i,x_j)\;.
\end{equation}
The corresponding cluster expansion takes the form
\begin{equation}\label{eq:rr31}
\ln \Xi_\Lambda(\nu)\;=\; \sum_{n\ge 1} \frac{1}{n!} \int_{\Lambda^n} \omega^T_n(x_1,\ldots, x_n)\,\nu(dx_1)\cdots \nu(dx_n)
\end{equation}
with truncated functions $\omega^T_n$ defined in \eqref{eq:rr32}. 

Our proof relies on the \emph{pinned cluster expansion}.  This is the function $\Gamma(\nu):S\to \mathbb C$ formally defined by
\begin{equation}\label{eq:rr41}
\boldGamma(\nu)(x)\;:=\; 1+\sum_{n\ge 0} \frac{1}{n!} \int_{S^n} \omega^T_{n+1}(x,x_1,\ldots, x_n)\,\nu(dx_1)\cdots \nu(dx_n)\;.
\end{equation}
This series is term-by-term bounded by the \emph{absolute pinned cluster expansion}
\begin{equation}\label{eq:rr61.ex}
\card{\boldGamma}\!(\nu)(x) \;:=\; 
1+\sum_{n\ge 0} \frac{1}{n!} \int_{S^n} \card{\omega^T_{n+1}(x,x_1,\ldots, x_n)}\,\absnu(dx_1)\cdots \absnu(dx_n)
\end{equation}
This absolute expansion will be  controlled through  the generalized tree sums, which
in our generalized setup take the form
\begin{equation}\label{eq:rr41.1}
\boldBeta(\nu)(x)\;:=\; 1+\sum_{n\ge 1} \frac{1}{n!} \int_{S^n} B_{n}(x;x_1,\ldots, x_n)\,\absnu(dx_1)\cdots \absnu(dx_n)\;.
\end{equation}
with the integrands $B_n$ defined as in \eqref{eq:rn5.gen.1} but with 
\begin{equation}\label{eq:rn5.gen}
\varphi_n(x;x_1,\ldots,x_n)\;=\; 
\prod_{1\le i\le n} \bigl|w(x,x_i)-1\bigr| \prod_{1\le i<j\le n} \card{w(x_i,x_j)}\;.
\end{equation}

The main cluster expansion result of our paper is the following theorem, proven in Section \ref{sec:cluster-proofs} below.

 \begin{theorem}\label{th:cluster.gen}
Consider a triple $(S,d\nu,w)$ as above such that
\begin{equation}\label{eq:rr.17}
\int \card{1-w(x,y)}\, \absnu(dy)\;<\; \infty\;.
\end{equation}
Define a map $\boldPsi:[0,\infty]^S\to [0,\infty]^S$ such that to each non-negative function $\boldmu: S\to [0,\infty]$ associates the function $\boldPsi(\boldmu): S\to [0,\infty]$ defined by
\begin{equation}\label{eq:rr50}
\boldPsi(\boldmu)(x)\;=\; 1 + \sum_{n\ge 1} \frac{1}{n!}  \int\limits_{ S^n} \varphi_n(x;x_1,\ldots, x_n)\,\mu(x_1)\cdots\mu(x_n)\,\absnu(dx_1)\cdots \absnu(dx_n) \;.
\end{equation}
If there exists a function $\boldmu$ such that, pointwisely,
\begin{equation}\label{eq:rr51}
\boldPsi(\boldmu)\;\le\; \boldmu\;,
\end{equation}
the following holds.
\begin{itemize}
\item[(i)] The functions $\boldPsi^n(\boldmu)$ belong to $[0,\infty]^S$ for each $n\ge 1$ and form a  pointwisely decreasing sequence that converges to a function $\boldPsi^\infty(\boldmu)\in  [0,\infty]^S$ which is a fix point of $\boldPsi$.  

\item[(ii)]
The tree sum function defined by \eqref{eq:rr41.1} is pointwise finite and satisfies the following family of pointwise bounds: 
\begin{equation}\label{eq:rr61.0}
\boldBeta(\nu)\;\le\; \boldPsi^\infty(\boldmu)\;\le\;\cdots\;\le\; \boldPsi^n(\boldmu)
\;\le\;\cdots\;\le\; \boldPsi(\boldmu)\;\le\;\boldmu\;.
\end{equation}

\end{itemize}

\end{theorem}
This theorem must be combined with the bounds summarized in the following proposition.
\begin{proposition}\label{prop:rr.1}
The following bounds hold:
\begin{equation}\label{eq:rr81.0}
\card{\boldGamma}\!(\nu)(x)\;\le\;  \boldBeta(\nu)(x)
\end{equation}
for each $x\in S$, and
\begin{equation}\label{eq:rr81.01}
\card{\ln \Xi_\Lambda(\nu)}\;\le\;  \int_\Lambda \card{\boldGamma}\!(\nu)(x)\,\absnu(dx)\;.
\end{equation}
These bounds are, in principle, understood as term-by-term in the sense of power series.  Upon convergence ---for instance under condition \eqref{eq:rr51}--- they become actual bounds between functions. 
\end{proposition}
\noindent 
The convergence of the generalized cluster expansions is an immediate corollary of Proposition \ref{prop:rr.1} and Theorem \ref{th:cluster.gen}. 

\begin{corollary}\label{cor:cluster.gen}
If condition \eqref{eq:rr51} is satisfied, then
\begin{itemize}
\item[(i)]
The pinned expansion \eqref{eq:rr41} converges term-by-term absolutely and is bounded above by $\boldmu$:  
\begin{equation}\label{eq:rr61}
\card{\boldGamma(\nu)} \;\le\; \card{\boldGamma}\!(\nu)\;\le\; \boldmu
\end{equation}
for each $x\in S$.
\item[(ii)] The cluster expansions \eqref{eq:rr31} converge absolutely and are bounded in the form
\begin{eqnarray}\label{eq:rr81}
\card{\ln \Xi_\Lambda(\nu)}&\le& \sum_{n\ge 1} \frac{1}{n!} \int_{\Lambda^n} \card{\omega^T_n(x_1,\ldots, x_n)}\,\absnu(dx_1)\cdots \absnu(dx_n)\nonumber\\
&\le& \int_\Lambda \boldmu\,d\!\absnu\;.
\end{eqnarray}
\end{itemize}
\end{corollary}

\begin{remark}
Condition \eqref{eq:rr51} strengthens preexisting criteria \cite{uel04,poguel09} which follow from ours by bounding $\card{w(x_i,x_j)}\le 1$ in \eqref{eq:rn5.gen}.  These criteria, however, apply also if the  weights correspond to stable but not necessarily positive interactions.  A strengthening for these general case is provided by the criterion in \cite{proyuh17}, conveniently adapted to the general framework \eqref{eq:rr30.1}-\eqref{eq:rr31}.  This strengthening, however, coincides with that in \cite{uel04,poguel09} for positive interactions.  
\end{remark}

\begin{remark}
It is natural to consider $\mu$ that is absolutely continuous with respect to some ``natural" no-negative measure $dx$ (e.g. counting or Lebesgue).  This means 
\begin{equation}\label{eq:rr100}
\nu(dx)\;=\;\xi(x)\,dx
\end{equation}
where the function $\xi(x)$ can be interpreted as a site-dependent (possibly complex) fugacity.  In this case 
\begin{eqnarray}\label{eq:rr70}
\boldPsi(\boldmu)(x)&=&1 + \sum_{n\ge 1} \frac{1}{n!}  \int\limits_{ S^n} \varphi_n(x;x_1,\ldots, x_n)\,\mu(x_1)\card{\xi(x_1)}\cdots\mu(x_n)\card{\xi(x_n)}\,d\,x_1\cdots d\,x_n \nonumber\\
&=& 1 + \sum_{n\ge 1} \frac{1}{n!}  \int\limits_{S^n} \varphi_n(x;x_1,\ldots, x_n)\,\widehat \mu(x_1)\cdots \widehat\mu(x_n)\,d\,x_1\cdots d\,x_n \\
&=:& \widehat \boldPsi(\widehat{\boldmu})(x)\;.\nonumber
\end{eqnarray}
Condition \eqref{eq:rr51} then translates into the existence of a function $\widehat\mu:S\to[0,\infty)$ such that
\begin{equation}\label{eq:rr71}
\card{\xi(x)}\,\widehat \boldPsi(\widehat \boldmu)(x)\;\le\; \widehat\mu(x) 
\end{equation}
for all $x\in S$.  This is a a generalized version of our initial condition \eqref{eq:rr9.0} when the interaction is not assumed to be translation invariant.  It determines a polydisc-like domain of fugacities for which the cluster expansions converge.  
\end{remark}

\begin{remark}\label{rem:r.exam}
The bound \eqref{eq:rr61.0} yields that condition \eqref{eq:rr71} implies the inequality
\begin{equation}\label{eq:rr61.g}
\card{\xi(x)}\,\widehat\boldBeta(\card{\boldxi})(x)\;\le\; \widehat\mu(x) 
\end{equation}
for $x\in S$, with
\begin{equation}\label{eq:rr41.1.g}
\widehat\boldBeta(\card{\boldxi})(x)\;:=\; 1+\sum_{n\ge 1} \frac{1}{n!} \int_{S^n} B_{n}(x;x_1,\ldots, x_n)\, \card{\xi(x_1)}\cdots \card{\xi(x_n)}\,d\,x_1\cdots d\,x_n\;.
\end{equation}
In particular \eqref{eq:rr71} guarantees the pointwise convergence of $\widehat\boldBeta(\card{\boldxi})$.
\end{remark}

\begin{remarks}\ 

\begin{itemize}
\item[(i)] With the choice \eqref{eq:rr100}, the function $\xi(x)\boldGamma(x)$ [see \eqref{eq:rr41}] corresponds to the \emph{one-point correlation density} of a gas with fugacity $\xi$.

\item[(ii)] 
In the translation invariant case it is natural to look for constant functions $\widehat\mu(x)=\mu$, in which case the results of the present section readily imply those of Section \ref{ssec:rr.10}.
\end{itemize}
\end{remarks}

\section{Results on virial expansions}
\subsection{General convergence result}\label{s:rr3}

The following theorem summarizes our main results on the convergence of the virial expansion. 

\begin{theorem}\label{th:rr1}
Denote $r^{**}$ the radius of convergence of $\Beta(z)$.  Then, the general term of the virial expansion for a repulsive gas is bounded above in the form
\begin{equation}\label{eq:rr200}
\frac{\card{\beta_{n+1}}}{(n+1)!} \;\le\; \frac{1}{n+1} \biggl[ \inf_{0\le r\le r^{**}}\Bigl(\frac 2r-\frac{1}{r\,\Beta(r)}\Bigr)\biggr]^n\;.
\end{equation}
Hence, the radius of convergence of the virial expansion is bounded below by
\begin{equation}\label{eq:rr201}
\mathcal{R}_{\mathrm{Vir}}\;\ge\; R^*\;:=\;
 \sup_{0\le r\le r^{**}}\,
\frac{r \Beta(r) }{2 \Beta(r)-1}\;.
\end{equation}
\end{theorem}

\begin{remark}\label{rem:radius}
Remark \ref{rem:r.exam} implies that
\begin{equation}\label{eq:rr202}
r^{**}\;\ge\; r^*\;=\;\max_{\mu\ge0} \frac{\mu}{\Psi(\mu)}
\end{equation}
with $\Psi(\mu)$ defined in \eqref{eq:rn8.0}.
\end{remark}

\begin{remark}
As $\beta$ increases when the repulsion decreases, weaker bounds are obtained by considering interactions $0\le\widetilde\phi\le \phi$.  That is, $R^*\ge \widetilde R^*$ where the latter is computed using the function $\widetilde\Beta(r)\ge \Beta(r)$ and the corresponding radii of convergence $\widetilde r^*$ and $\widetilde r^{**}$. 
\end{remark}

\begin{remark}
The classical convergence bound, due to Lebowitz and Penrose, in the case of positive interactions is
\begin{equation}\label{eq:rr203}
R_{\mathrm{LP}}\;=\;
\frac{0.144766998}{C(\beta)}\;.
\end{equation}
To compare it with \eqref{eq:rr201} we resort to the following identity proven in Section \ref{sec:pf.virial} below:
\begin{equation}\label{eq:rr204}
\frac 2r-\frac{1}{r\,\Beta(r)}\;=\;\frac{1+  T^{\mathrm{Pen,1}}(r)}{r}
\end{equation}
where $T^{\mathrm{Pen,1}}(r)$ is a power series defined as $\Beta(r)$ but involving less trees, namely the so-called \emph{Penrose-unsplitable trees} defined in Section \ref{sec:unspit} below.    
\begin{equation}\label{eq:rr205}
T^{\mathrm{Pen,1}}(r)\;:=\; 1+\sum_{n\ge 1} \frac{r^n}{n!} \int_{\mathbb R^{dn}} B^{\mathrm{Pen},1}_n(0;x_1,\ldots, x_n)\,dx_1\cdots dx_n\;.
\end{equation}
with
\begin{equation}\label{eq:rr206}
B^{\mathrm{Pen},1}_n(0;x_1,\ldots,x_n)\;=\; \sum_{\tau\in\mathcal T^0_{\mathrm{Pen},1}[n]} \,\prod_{i=0}^{n}
\varphi_{s_i(\tau)}(x_i;x_{(i,1)},\ldots,x_{(i,s_i(\tau))})\;.
\end{equation}
where $\mathcal T^0_{\mathrm{Pen},1}[n]$ denotes the family of Penrose-unsplittable trees rooted in 0 with $n$ non-root vertices.  These trees are properly defined in Section \ref{sec:unspit} below, but for the present purposes we need only two of their properties:
\begin{itemize}
\item[(i)] The $n$-th term of the series \eqref{eq:rr205} satisfies the bound
\begin{equation}\label{eq:rr207}
n![r^n]\,T^{\mathrm{Pen},1}(r)\;\le\; C(\beta)^n\card{\mathcal{T}^0_{\mathrm{Pen},1}[n]}\;.
\end{equation}
[We use the standard formal series notation reviewed in Section \ref{ss:fseroper} below.]
\item[(ii)] The following precise counting result~\cite{ramtat15}:
\begin{equation}\label{eq:rr208}
\card{\mathcal{T}^0_{\mathrm{Pen},1}[n]} \;=\; (n-1)^{(n-1)}\;.
\end{equation}
\end{itemize}
From \eqref{eq:rr204}--\eqref{eq:rr208} we recover the following criterion proven in~\cite{ramtat15} which coincides with the announcement given in \cite[Chapter IV, Section 3.7]{gro} .
\begin{theorem}{\bf (Groenveld-Ramawadth-Tate)}\label{theorem1}
Let 
\begin{equation}\label{eq:rr209}
T_1(r)\;=\; r+ \sum_{n\ge 1} n^n\,\frac{r^{n+1}}{(n+1)!}\;.
\end{equation}
Then 
\begin{equation}\label{eq:rr210}
\frac{\card{\beta_{n+1}}}{(n+1)!} \;\le\; \frac{C(\beta)^n}{n+1} \Bigl(\frac{1+T_1(\widehat r)}{\widehat r}\Bigr)^n
\end{equation}
where $\widehat r$ is the smallest solution of
\begin{equation}\label{eq:rr211}
r\,T_1'(r)-T_1(r)\;=\;1\;.
\end{equation}
As a consequence,
\begin{equation}\label{eq:rr212}
\mathcal{R}_{\mathrm{Vir}}\;\ge\; R_{\mathrm{GRT}}\;:=\;
\frac{1}{C(\beta)}\,\Bigl(\frac{\widehat r}{1+T_1(\widehat r)}\Bigr)\;.
\end{equation}
\end{theorem} 
Numerical estimations yield \cite{ramtat15}
\begin{equation}\label{eq:rr213}
R_{\mathrm{GRT}}\;=\;\frac{0.237961}{C(\beta)}
\end{equation}
which is more than 60\% larger than the classical estimate \eqref{eq:rr203}.
\end{remark}

\subsection{Efficient convergence criterion}\label{ssec:virial.eff}
The bound \eqref{eq:rr200} is not easy to estimate precisely.  In principle, it should bring improved estimations by exploiting the repulsion between siblings.  In doing so, however, one looses the precise counting \eqref{eq:rr208} and the net effect could be deceptive.  Here we derive a weaker criterion that in many cases succeeds in revealing improvements over the GRT bound.  The approach amounts to a change of variable combined with a not very drastic bound.  Let us relate each $0\le\mu\le\mu^*=\argmax_{\mu\ge 0} [\mu/\Psi(\mu)]$ with $r$ such that
\begin{equation}\label{eq:rr214}
r\,\Psi(\mu)\;=\;\mu\;.
\end{equation}
Such an $r$ always exists by observation (b) of Remark \ref{rem:rr10}.  As, by the
rightmost inequality in \eqref{eq:rr300}, $r \Beta(r)\le \mu$, we conclude that
\begin{equation}\label{eq:rr215}
\Beta(r)\;\le\; \Psi(\mu)\;.
\end{equation}
Expressions \eqref{eq:rr214} and \eqref{eq:rr215} lead to the following corollary of Theorem \ref{th:rr1}.

\begin{corollary}\label{cor:rr100}
The radius of convergence of the virial expansion is bounded below by
\begin{equation}\label{eq:rr221.1}
\mathcal{R}_{\mathrm{Vir}}\;\ge\; R^*\;\ge\;M^*\;:=\;
 \sup_{0\le \mu\le \mu^{*}}\,
\frac{\mu }{2 \Psi(\mu)-1}\;.
\end{equation}
\end{corollary}

\begin{remark}\label{rem:rr221}
As before, weaker bounds can be obtained considering a less repulsive potential $\widetilde\phi$.  The corresponding bound involves the larger function $\widetilde \Psi$:
\begin{equation}\label{eq:rr221}
M^*\;\ge\;\widetilde M^*\;:=\;
 \sup_{0\le \mu\le \widetilde\mu^{*}}\,
\frac{\mu }{2 \widetilde\Psi(\mu)-1}
\qquad,\quad \widetilde\mu^*:=\argmax_{\mu\ge 0} \frac{\mu}{\widetilde\Psi(\mu)}\;.
\end{equation}
\end{remark}
We now present a few applications of this criterion.

\begin{application}[{\bf Groenveld-Ramawadth-Tate almost recovered}]
An initial test is to neglect the repulsion between siblings and consider the coarser bound $g(n)\le C(\beta)^n$.  This corresponds to the bound \eqref{eq:rr207} and yields $\widetilde\Psi(\mu)=\exp[C(\beta)\mu]$.  Formula \eqref{eq:rr221}  yields
\begin{equation}\label{eq:rr230}
M_1\;=\; \frac{1}{C(\beta)}\,\max_{0\le a\le 1} \frac{\alpha}{2\eee^\alpha -1}\;
\end{equation}
A short calculation shows that the maximum is attained at $\alpha^*$ solution of the equation 
$\eee^{\alpha^*}(1-\alpha^*)=1/2$.  To six decimal places we have $\alpha^*=0.768039$.  This yields
\begin{equation}\label{eq:rr230.1}
M_1\;=\; \frac{1}{C(\beta)}\,(1-\alpha^*)\;=\;\frac{0.231961}{C(\beta)}
\end{equation}
less than 3\% smaller than $R_{\mathrm{GRT}}$.  This proves that, despite not using the precise counting \eqref{eq:rr208}, the estimation in Corollary \ref{cor:rr100} is indeed very efficient.
\end{application}

\begin{application}[{\bf Hard spheres}]
For the hard-core models introduced in Application \ref{app.hard}, condition \eqref{eq:rr221.1} becomes
\begin{equation}\label{eq:rr221.1.hard}
\mathcal{R}_{\mathrm{Vir}}\;\ge\; M_d^{\mathrm{hard}}\;:=\;
 \frac{1}{V_d(a)}\,\sup_{0\le \alpha\le \alpha_d^{*}}\,
\frac{\alpha }{2 \widehat\Psi_d(\alpha)-1}
\qquad,\quad \alpha_d^*=\argmax_{\mu\ge 0}\frac{\mu}{\widehat\Psi_d(\mu)}
\end{equation}

For $d=1$ (hard intervals in a line) 
\begin{equation}\label{eq:rr221.2.hard}
M_1^{\mathrm{hard}}\;=\; \frac{1}{V_1(1)}\,\max_{0\le\alpha\le \sqrt 8}
\frac{\alpha }{1+2\alpha+\frac14 \alpha^2}\;=\;\frac{1/3}{V_1(a)}\;.
\end{equation}
which is almost 40\% larger than $R_{\mathrm{GRT}}$ (and 130\% larger than the classical estimate $R_{\mathrm{LP}}$).

For $d=2$ (hard disks), the normalized vertex function is \cite{fps}
\begin{equation}\label{eq:rr221.4.hard}
\widehat\Psi_2(\alpha)\;=\;1+\alpha+\frac{3\sqrt{3}}{8\pi}\alpha^2+\frac{0.0589}{6}\alpha^3+\frac{0.00013}{24}a^4+\frac{0.0001}{120}\alpha^5\;.
 \end{equation}
A numeric computation yields 
\begin{equation}\label{eq:rr221.4.hard.1}
M_2^{\mathrm{hard}}\;=\;\frac{0.300224}{V_2(1)}
\end{equation}
25\% larger than $R_{\mathrm{GRT}}$.
\end{application}

\begin{application}[{\bf Power-law potentials}]
The virial coefficients of the models of Application \ref{app.power} were numerically determined in \cite{Barlow12} and  radii of convergence estimated through appropriate extrapolation procedures.  The authors recognize that the convergence issue is complicated by the different signs of the virial coefficients.  We hope that rigorous results, such as ours, help to dispel doubts associated to numerical procedures.  As a contribution in this direction, we contrast here convergence radii found with the formulas above with the numerical estimations for a few cases.  For simplicity we consider $\beta=1$ in all our computations.  

For starters we compare $R_{GRT}$ with the numerical convergence radii given in \cite[Table VII]{Barlow12} for three dimensional models with $\epsilon=\sigma=1$ for different vaules of $n$.  This comparison is summarized in Table \ref{tab:rr1}. In the present case, $R_{GRT}=0.237961[C(1)]^{-1}$ with
\begin{equation}\label{eq:rr1.1power}
C(1)\;=\;4\pi\int\limits_{r>0}\left[1-\eee^{-(\frac{1}{r})^n}\right]r^2dr\;.
\end{equation}
\begin{table}[h]
\begin{tabular}{|c|c|c|}
\hline
$n$ & $R_{GRT}$ & $R_{\mathrm{Num}}$\\
\hline
4 & 0.0153 & 0.1092\\
5 & 0.025 & 0.2418\\
6 & 0.0312 & 0.3280\\
7 & 0.0355 & 0.4022\\
8 & 0.0386 & 0.4634\\
\hline
\end{tabular}
\caption{GRT bounds on radii of convergence vs numerical estimations \cite{Barlow12} for potentials \eqref{eq:rr1.power} with $\beta=\epsilon=\sigma=1$ }\label{tab:rr1}
\end{table}
The table shows differences increasing with $n$, with numerical estimations being larger by a factor ranging from 7 to 12.  This trend is to be expected, given that the repulsion effect neglected in the GRT-approach increases with $n$.  

To measure the effect of this extra repulsion we applied our approach to the case $n=6$ using the bound on the function $\widehat\Psi$ explained in Application \ref{app.power}.  
%
%
%
We obtain $\mathcal{R}_{\rm{Vir}}\ge 0.2612494[C(1)]^{-1}\\\;=\;0.035$. The difference with the numerical value given in Table \ref{tab:rr1} is a factor 9.5 which improves the factor 10.5 resulting from the GRT approach.

\end{application}

\section{Proof of the cluster expansion results}\label{sec:cluster-proofs}

In this section we turn to the proof of the results of Sections \ref{ssec:cluster.general}, which in turns imply those of the preceding sections. We need only to prove Theorem \ref{th:cluster.gen} and Proposition \ref{prop:rr.1}.

\subsection{Preliminary bounds: Partition schemes}\label{ssec:cluster-partitions}

Expression \eqref{eq:rr32} is well known.  It can be obtained by expanding $w(x_i,x_j)=[w(x_i,x_j)-1]+1$ and using that the weights $\omega_n^T$ are the only solution to the equations
\begin{equation}\label{eq:rrp1}
\omega_n\;=\;\sum_{k=1}^n \sum_{\{I_1,\cdots, I_k\} \;\mathrm{ n.t.part. }\; [n]} \omega^T_{\card{I_1}}\cdots \omega^T_{\card{I_k}}
\end{equation}
where ``$\mathrm{ n.t.part. }\; [n]$" means \emph{non-trivial partition of} $\{1,\ldots,n\}$, that is a partition with no empty component.  The more conomical form of checking the validity of expansion \eqref{eq:rr31} ---in the sense of power series and then making $z=1$--- is, perhaps, through the algebraic approach of \cite[Chapter 4]{Rue69}.  Convergence studies are based on rewriting \eqref{eq:rr32} in terms of trees.  The optimal way to organize these rewriting is in terms of \emph{partition schemes} \cite{scosok05,Fern}.

Let us consider the set $\mathcal{C}[n]$ of graphs with vertices $\{1,\ldots, n\}$ partial ordered by bond inclusion:
\begin{equation}\label{eq:rr.1.part}
\mathrm{g}\le \mathrm{g}'\quad\Longleftrightarrow\quad E(\mathrm{g})\subset E(\mathrm{g}')\;.
\end{equation}
If $\mathrm{g}\le \mathrm{g}'$, let us denote $[\mathrm{g},\mathrm{g}']$ the set of $\tilde{\mathrm{g}}\in\mathcal{C}[n]$ such that $\mathrm{g}\le \tilde{\mathrm{g}}\le \mathrm{g}'$.  Let us denote $\mathcal{T}[n]\subset\mathcal{C}[n]$ the subset of those graphs that are trees.

\begin{definition}[{\bf Partition scheme}] Let us call a \emph{partition scheme} for the family $\mathcal{C}[n]$ any map $R:\mathcal{T}[n]\to\mathcal{C}[n]:\tau\mapsto R(\tau)$ such that
\begin{itemize}
\item[(i)] $E(R(\tau))\supset E(\tau)$, and 
\item[(ii)] $\mathcal{C}[n]$ is the disjoint union of the sets $[\tau,R(\tau)]$, $\tau\in\trees[n]$.
\end{itemize}
\end{definition}

\begin{proposition}\label{prop:rr100.part}
For any partition scheme $R$,
\begin{equation}\label{eq:rr.2.part}
\omega^T_n(x_1,\ldots, x_n)\;=\; \sum_{\tau\in\trees[n]}\prod_{\{i,j\}\in E(\tau)}\bigl[w(x_i,x_j)-1\bigr]\prod_{\{i,j\}\in E(R(\tau))\setminus E(\tau)}w(x_i,x_j)\;.
\end{equation}
\end{proposition}
\begin{proof}
The proof of this proposition is well known \cite{scosok05}, but let us reproduce it for the sake of completeness.  Let us denote $f(i,j) :=w(x_i,x_j)-1$.  Then, 
\begin{eqnarray}
\sum_{\mathrm{g}\in\mathcal{C}[n+1]}\prod_{\{i,j\}\in E(g)}f(i,j)
&=&\sum_{\tau\in\trees[n+1]}\sum_{\mathrm{g}\in [\tau,R(\tau)]}\prod_{\{i,j\}\in E(g)}f(i,j)\nonumber\\
&=&\sum_{\tau\in\trees[n+1]}\prod_{\{i,j\}\in E(\tau)}f(i,j)
\sum_{\mathrm{g}\in [\tau,R(\tau)]}\prod_{\{i,j\}\in E(\mathrm{g})\setminus E(\tau)}f(i,j)\\
&=&\sum_{\tau\in\trees[n+1]}\prod_{\{i,j\}\in E(\tau)}f(i,j)
\prod_{\{i,j\}\in E(R(\tau))\setminus E(\tau)}\bigl[f(i,j)+1\bigr]\nonumber\;.
\end{eqnarray}
\end{proof}

Despite the generality of \eqref{eq:rr.2.part}, here we will be concerned with a very particular partition scheme, called the \emph{Penrose partition} \cite{pen67,Fern}.  The scheme has two advantages: (i) it leads to useful bounds purely in terms of vertices and their siblings, and (ii) it allows the \emph{faithful mergings} described below which are crucial for the improvements in the determination of virial convergence.  
For our purposes it is better to describe this partition considering the set $\mathcal{C}[n+1]$ of graphs with vertices $\{0,1,\ldots, n\}$. We think the corresponding trees $\mathcal{T}[n+1]$ as \emph{rooted} in the vertex 0.  To emphasize this fact we write $\mathcal{T}[n+1]=\mathcal{T}^0[n]$.  Trees have an associated notion of \emph{tree distance} between vertices given by the number of links needed to go from one the other through the tree.  Rooted trees, in addition, admit the notion of \emph{generation}.  For $\tau \in \mathcal{T}^0[n]$ the generation number of a vertex $i$ is its distance $d_{\tau}(i)$ to the root 0.  This number brings in associated notions of kinship: If $\{i,j\}\in E(\tau)$ and $d_{\tau}(i)=d_{\tau}(j)-1$, we say that $i$ is the \emph{predecessor} or \emph{parent} of $j$, and $j$ is the \emph{child} of $i$.  Children of the same parent are called \emph{siblings}. Other relations ---grandparents, uncles, etc--- can be analogously defined.

\begin{definition}[{\bf Penrose partition scheme}]
The partition is defined associating to each tree $\tau$ the graph $R_{\mathrm{Pen}}(\tau)$ obtained by adding to $\tau$ all edges $\{i,j\}\not\in E(\tau)$ such that either:
\begin{itemize}
\item[(i)] $d_{\tau}(i)=d_\tau(j)$ (edges between vertices at same generation,)
\item[(ii)] $d_\tau(j)=d_\tau(i)-1$ and $i'<j$, where $i'$ is the predecessor of $i$ (edges between a vertex and uncles with lower index than the parent).
\end{itemize}
\end{definition}
The proof that this prescription indeed constitutes a partition scheme is basically given in Penrose's original publication~\cite{pen67}. More formal proofs can be found in the references given in \cite[Section 2.2]{scosok05} and in \cite{temmel14}.

\subsection{Proof of Theorem \ref{th:cluster.gen}}\label{ssec:rrproof.gen}

\paragraph{\textbf{Proof of part (i)}} As $\boldPhi$ is defined through positive integrands and measures, condition \eqref{eq:rr51} implies inductively that $0\le\boldPsi^{n+1}\le \boldPsi^n$ for each $n\ge 1$.  This monotonicity  leads to the existence of the pointwise limit $\boldPsi^\infty$.  By monotone convergence ---Beppo Levi theorem for each integral coefficient and convergence of monotone series for the whole sum:
\begin{equation}\label{eq:rr500}
\boldPsi\bigl(\boldPsi^\infty\bigr)\;=\; \boldPsi\bigl(\lim_n \boldPsi^n\bigr)\;=\; \lim_n \boldPsi\bigl( \boldPsi^n\bigr)\;=\; \boldPsi^\infty\;. \qed
\end{equation}

\paragraph{\textbf{Proof of part (ii)}} 

The proof relies on a combinatorial identity expressing the fact that rooted trees can be generated by joining subtrees rooted in the children of the root, and this property is preserved if the trees have weights that depend only on siblings. One way to prove this property was presented in \cite{fps}.  Here we adopt a different approach discussed by one of us with Alan Sokal, Roman Koteck\'y and Daniel Ueltschi (unpublished, 2011) .  We need some definitions.  Let us call the \emph{depth} of a tree its maximal generation number and denote $\trees_N[n+1]$ the set of trees with vertices $\{0,1,\ldots,n\}$, rooted at 0 and with depth not exceeding $N$.  The key combinatorial result is the following.

\begin{proposition}[{\bf Iterative generation of trees}]\label{prop:proof.1}
Consider a space $S$ endowed with a positive measure $dx$, and measurable functions $c_n:S\times S^n\longrightarrow [0,\infty]$, $n\ge 1$, that are invariant under permutations of their last $n$ arguments (\emph{vertex weights}) and define $c_0\equiv 1$.  For each integer $N\ge 0$ consider the tree sums
\begin{equation}\label{eq:rr.keycomb.1}
C_n^N(x_0;x_1,\ldots,x_n)\;=\; \sum_{\tau\in\mathcal T_N[n+1]} \prod_{i=0}^{n}
c_{s_i}(x_i;x_{(i,1)},\ldots,x_{(i,s_i)})
\end{equation} 
and the functions $\boldCeta^N: S\to [0,\infty]$ defined by
\begin{equation}\label{eq:rr.keycomb.2}
\boldCeta^N( x)\;:=\; 1+\sum_{n\ge 1} \frac{1}{n!} \int_{S^n} C^N_{n}(x;x_1,\ldots, x_n)\,d\,x_1\cdots d\,x_n
\end{equation}
(\emph{integrated tree sums}).
Finally, define the \emph{tree-generating operator} $\boldFeta:[0,\infty]^S\to [0,\infty]^S$ such that to each non-negative function $\boldmu: S\to [0,\infty]$ associates the function $\boldFeta(\boldmu): S\to [0,\infty]$ defined by
\begin{equation}\label{eq:rr.keycomb.3}
\boldFeta(\boldmu)(x)\;=\; 1 + \sum_{n\ge 1} \frac{1}{n!}  \int_{ S^n} c_n(x;x_1,\ldots, x_n)\,\mu(x_1)\cdots\mu(x_n)\,dx_1\cdots dx_n \;.
\end{equation}
Then, 
\begin{eqnarray}\label{eq:rr.keycomb.4}
\boldFeta(\boldzero)&=&\boldone\nonumber\\
\boldFeta\bigl(\boldCeta^N\bigr)&=& \boldCeta^{N+1} \qquad, \quad N\ge 0\;.
\end{eqnarray}
Here  $\boldzero$ and $\boldone$ are the functions identically zero and 1 respectively, and $\boldCeta^0=\boldone$.  
\end{proposition}
\begin{proof}
Every tree in $C_n^{N+1}$ is uniquely defined by a partition $\{I_0,I_1,\ldots, I_{\card{I_0}}\}$ of $\{1,\ldots,n\}$ and a family of trees $\tau_i\in \trees_N[\card{I_i}]$.  The indices in $I_0$ are the children of the root 0 and each $\tau_i$ is the tree formed by all the descendants of the child $(0,i)$ of the root. Therefore, if $n, N\ge 1$,
\begin{equation}\label{eq:rr.keycomb.5}
C_n^{N+1}(x_0;x_1,\ldots,x_n)\;=\; \sum_{\{I_0,I_1,\ldots, I_{\card{I_0}}\}\atop \mathrm{ part. }\; [n]}
c_{\card{I_0}}\bigl(x_0; x_{I_0}\bigr) \prod_{i=1}^{\card{I_0}} C_n^{N}\bigl(x_{(0,i)};x_{I_i}\bigr)
\end{equation}
where we denote $x_I=(x_j)_{j\in I}$. Some of the sets $I_i$, $i\ge 1$, can be empty, in which case $C_n^{N}\bigl(x_{(0,i)};x_{\emptyset}\bigr)= 1$.
The permutation invariance of the vertex weights implies that in the integrated version of this last equation the terms on the right-hand side depend only on the cardinality of the sets in the partition.  Denoting
\begin{equation}\label{eq:rr.keycomb.6}
\overline C^M_n(x)\;:=\; \int_{S^n} C_n^{M}(x;x_1,\ldots,x_n)\,dx_1\cdots dx_n
\end{equation}
and integrating both sides of \eqref{eq:rr.keycomb.5} we have
\begin{equation}\label{eq:rr.keycomb.7}
\overline C_n^{N+1}(x_0)\;=\;
\sum_{k,\ell_1,\cdots,\ell_{k} \atop k+\ell_1+\cdots+\ell_{k}=n} 
\binom{n}{k\,\ell_1\cdots\,\ell_k}
\int_{S^k} c_k\bigl(x_0; x_1,\ldots, x_k\bigr) \prod_{i=1}^k \Bigl[\overline C_{\ell_i}^{N}(x_i)\, dx_i\Bigr]\;.
\end{equation}
The combinatorial factor equals the number of partitions $\{I_0,I_1,\ldots, I_{\card{I_0}}\}$ of $\{1,\ldots, n\}$ with $\card{I_0}=k $ and $\card{I_i}=\ell_i$ for $1\le i\le k$.  Dividing the right-hand side by $n!$ and summing over $n$ we finally get
\begin{eqnarray}\label{eq:rr.keycomb.8}
\lefteqn{\boldCeta^{N+1}(x_0)}\nonumber\\
 &=& 1+ \sum_{k\ge 1} \frac{1}{k!}\int_{S^k} \!dx_1\cdots dx_n \,c_k\bigl(x_0; x_1,\ldots, x_k\bigr) \prod_{i=1}^k \Bigl[\sum_{\ell_i\ge 0}\frac{1}{\ell_i!}\overline C_{\ell_i}^{N}(x_i)\Bigr] \nonumber\\
 &=& \boldFeta\bigl(\boldCeta^N\bigr)\;.
\end{eqnarray}
\end{proof}

Let us now conclude the proof of Part (ii) of Theorem \ref{th:cluster.gen}.  By proposition \ref{prop:proof.1} we have that
\begin{equation}\label{eq:rr.keycomb.9}
\boldBeta^N\;=\boldPsi^{N+1}(\boldzero)
\end{equation}
for each $N\ge 0$.  Here $\boldBeta^N$ is defined as in \eqref{eq:rr41.1} but summing only over trees of depth at most $N$.  As all the relevant coefficients are positive, $\boldBeta^N$ increases pointwisely with $N$ and, by monotonic convergence, converges to $\boldBeta$.  Furthermore, 
\begin{equation}\label{eq:rr.keycomb.10}
\boldPsi^{N}(\boldzero) \;\le\; \boldPsi^{N}(\boldmu)\;.
\end{equation}
The following line summarizes the consequences of inequalities \eqref{eq:rr.keycomb.9}, \eqref{eq:rr.keycomb.10} and Part (i).  In particular it proves \eqref{eq:rr61.0}.
\begin{equation}\label{eq:rr.keycomb.11}
\boldzero\;\le\;\boldPsi^{N}(\boldzero) \;\nearrow_{N\to\infty}\; \boldBeta\;\le\;
\boldPsi^{\infty}(\boldmu)\;\swarrow_{N\to\infty}\boldPsi^{N}(\boldmu)\;.\le\;\boldmu\;. \qed
\end{equation}

\subsection{Proof of Proposition \ref{prop:rr.1}}
To prove inequality \eqref{eq:rr81.0} we start from the bound
\begin{equation}\label{eq:rr.100.proof}
\card{\omega^T_n(x_1,\ldots, x_n)}\;\le \; \sum_{\tau\in\trees[n]}\prod_{\{i,j\}\in E(\tau)}\card{w(x_i,x_j)-1}\prod_{\{i,j\}\in E(R_{\mathrm{Pen}}(\tau))\setminus E(\tau)} \card{w(x_i,x_j)}\;.
\end{equation}
which follows from Proposition \ref{prop:rr100.part}.  Next, for each tree $\tau$ we bound by 1 all edges in $E(R_{\mathrm{Pen}}(\tau))$ except those between siblings.  This yields the inequality
\begin{equation}\label{eq:rr.101.proof}
\card{\omega^T_{n+1}(x,x_1,\ldots, x_n)} \;\le\; B_{n}(x;x_1,\ldots, x_n)
\end{equation}
proving the term-by-term majorization of $\card{\boldGamma}\!(\nu)(x)$ by $\boldBeta(\nu)(x)$ for every $x\in S$.
 
Inequality \eqref{eq:rr81.01} is 
a simple consequence of the inequalities
\begin{eqnarray*}
\lefteqn{\sum_{n\ge 1} \frac{1}{n!} \int_{\Lambda^n} \card{\omega^T_n(x_1,\ldots, x_n)}\,\absnu(dx_1)\cdots \absnu(dx_n)}\\
&\le& \sum_{n\ge 1} \frac{1}{n!} \int_{S^n} \card{\omega^T_n(x_1,\ldots, x_n)}\,\mathbb{\Large 1}[x_i\in\Lambda \mbox{ for some } 1\le i\le n]\, \absnu(dx_1)\cdots \absnu(dx_n)\nonumber\\
&\le& \sum_{n\ge 1} \frac{1}{n!} \sum_{i=1}^n\int_{S^n} \card{\omega^T_n(x_1,\ldots, x_n)}\,\mathbb{\Large 1}[x_i\in\Lambda]\, \absnu(dx_1)\cdots \absnu(dx_n)\nonumber\;.
\end{eqnarray*}
Due to the  invariance of $\omega^T_n$ under permutation of arguments we conclude that the last expression is no smaller than
\[
\sum_{n\ge 1} \frac{1}{n!} \,n\int_\Lambda \absnu(dx_1)\int_{S^{n-1}} \card{\omega^T_n(x_1,\ldots, x_n)}\,\, \absnu(dx_2)\cdots \absnu(dx_n)
\;=\; \int_\Lambda \absnu(dx_1) \card{\boldGamma}\!(\nu)(x)\;. \qed
\]

\subsection{Simpler proof of (almost) criterion \eqref{eq:rr9.0}} \label{ssec:rrproof.eas}
The convergence criterion \eqref{eq:rr9.0} ---albeit in a non-sharp version--- admits an alternative proof relying only on Cayley formula for the number of trees with given degrees.  This proof is basically contained in the approach used in \cite{fps} for the cluster expansion for hard spheres.  The following proposition is the main step of this proof.  It applies for translation-invariant interactions and involves the tree sum $\Beta(z)$, defined in \eqref{eq:rr41.1.gg}, and the vertex sum $\Psi(\mu)$ defined in \eqref{eq:rn8.0}.
The relation with the actual cluster expansions follows from the inequalities 
\begin{equation}\label{eq:rr9.-1.easy}
\card{\beta P_\Lambda(z)},\,  \card{\beta p(z)} \;\le\; \card z\,\Beta(\card z)\;,
\end{equation}
which are a consequence of Proposition \ref{prop:rr100.part} [see \eqref{eq:rr.101.proof}]. 

\begin{proposition}\label{prop:rr.easy}
If $z\in\mathbb C$ is such that 
\begin{equation}\label{eq:rr9.0.easy}
\card z\;<\; \frac{\mu}{\Psi({\mu})}
\end{equation}
for some $\mu>0$, then $\Beta(z)$ converges absolutely and, furthermore,  
\begin{equation}\label{eq:rr9.0.01.easy}
 \Beta(\card z)\;\le\; \frac{1}{1- \card z\frac{\Psi(\mu)}{\mu}}\;.
\end{equation}
\end{proposition}

\begin{remark}
The bound \eqref{eq:rr9.0.01.easy} implies that in the smaller region
\begin{equation}\label{eq:rr9.10.easy}
\card z\;<\; \frac{\mu}{1+\Psi({\mu})}
\end{equation}
we have
\begin{equation}\label{eq:rr9.11.easy}
\card z \Beta(\card z)\;\le\;\mu\;,
\end{equation}
which, together with \eqref{eq:rr9.-1.easy} (almost) proves Part (ii) of Theorem \ref{th:cluster}
\end{remark}

\begin{proof}
Let us denote
\begin{equation}\label{eq:rr1.proof.easy}
\overline B_n \;=\;\int_{\mathbb R^{dn}} B_{n}(0;x_1,\ldots, x_n)\,d\,x_1\cdots d\,x_n\;.
\end{equation}
The factorization structure of \eqref{eq:rn5.gen.1} together with the translation invariance imply than
\begin{equation}\label{eq:rr2.proof.easy}
\overline B_n \;=\;\sum_{\tau\in\mathcal T^0[n]} \prod_{i=0}^{n} g(s_i)
\end{equation}
with $\tg(n)$ defined in \eqref{eq:rr10}.  Denoting $d_i$ the degree of each vertex, we observe that $d_0=s_0$ but $d_i=s_i+1$ for $i=1,\ldots, n$.  Thus,
\begin{equation}\label{eq:rr3.proof.easy}
\overline B_n \;=\;\sum_{d_0\ge 0, d_1,\ldots,d_n\ge 1 \atop d_0+\cdots+d_n=2n} N(d_0,\ldots,d_n)\, g(d_0)\prod_{i=1}^{n} g(d_i-1)
\end{equation}
The numeric factor is the number of trees with $n+1$ vertices and the specified degrees.  By Cayley's theorem it equals
\begin{equation}\label{eq:cayley}
N(d_0,\ldots,d_n)\;=\; \binom{n-1}{(d_0-1)\,\cdots\, (d_n-1)}\;.
\end{equation}
Hence,
\begin{eqnarray}\label{eq:rr4.proof.easy}
\frac{{\card z}^n}{n!}\,\overline B_n &=& \frac{{\card z}^n}{n}\,\sum_{d_0\ge 0, d_1,\ldots,d_n\ge 1 \atop d_0+\cdots+d_n=2n} \frac{g(d_0)}{(d_0-1)!}\prod_{i=1}^{n} \frac{g(d_i-1)}{(d_i-1)!}\nonumber\\
&=&\left(\frac{{\card z}}{\mu}\right)^{n}\,\sum_{d_0\ge 0, d_1,\ldots,d_n\ge 1 \atop d_0+\cdots+d_n=2n} \frac{g(d_0)\,\mu^{d_0}}{n(d_0-1)!}\prod_{i=1}^{n} \frac{g(d_i-1)\,\mu^{d_i-1}}{(d_i-1)!}\;.
\end{eqnarray}
The last line simply amounts to multiplying and dividing the right-hand side by $\mu^n$.  We now bound $1/[n(d_0-1)!]$ by $1/d_0!$ and sum over $n$ to conclude
\begin{equation}\label{eq:rr5.proof.easy}
\Beta(\card z)\;\le\; \sum_{n\ge 0} \left[\frac{\card z\,\Pi(\mu)}{\mu}\right]^n\;=\;\frac{1}{1- \card z\frac{\Psi(\mu)}{\mu}}\;.
\end{equation}
\end{proof}

\section{Expression of the virial coefficients}\label{sec:pf.virial-coeff}
In this section we present a detailed derivation of the alternative expression for the virial coefficients to be used in the proof of Theorem \ref{th:rr1}.

\subsection{Formal power series and their operations}\label{ss:fseroper}
The rigorous study of expansion techniques in Physics is a two-step process: (1) Find the general expression ---and some properties--- of the coefficients.  (2) Determine the convergence properties of the resulting series.  The combinatorial aspects inherent to the first step are part of the theory of formal power series.  As this theory may not be well known by the intended readership of the present article, we present in the sequel a brief review, following the excellent summary in \cite{wiki_formal}.

Informally, formal power series are generalizations of polynomials involving countably many powers.  Rigorously, a \emph{formal power series on a commutative ring} $R$ (here $R=\mathbb{R}$) is a countable family of coefficients $\underline a=(a_n)_{n\in\mathbb{N}}$ endowed with the ring structure provided by the coordinate-wise sum and the convolution product:
\begin{eqnarray}\label{eq:bell.1}
\underline a + \underline b&=&(a_n+b_n)\\
\label{eq:bell.2}
(\underline a * \underline b)_n &=& \sum_{k=0}^n a_k\,b_{n-k}\;.
\end{eqnarray}
An alternative notation showing more clearly their character of generalized polynomials is obtained by denoting $X=(\delta_{n,1})$ and, more generally, $X^m=\underline\delta^{*m}=(\delta_{n,m})$.  Using the previous operations we obtain
\begin{equation}\label{eq:bell.3}
\underline a\;=\; \sum_{n\ge 0} a_n\,X^n\;=:\; \underline a(X)
\end{equation}
and the operations \eqref{eq:bell.1}-\eqref{eq:bell.2} map into the natural generalizations of the sum and product of polynomials:
\begin{eqnarray}\label{eq:bell.1.1}
\underline a(X) + \underline b(X)&=&\sum_{n\ge 0}(a_n+b_n)X^n\\
\label{eq:bell.2.1}
\underline a(X) \, \underline b(X)&=&\sum_{n\ge 0}(\underline a * \underline b)_n X^n\;.
\end{eqnarray}
The resulting ring of series is denoted $\rser$.  It is endowed with the product discrete topology topology of $R^{\mathbb N}$, in which convergence means subsequent stabilization of terms.  In this way the sum in \eqref{eq:bell.3} becomes an honest infinite sum. If the object $X$ is allowed to take values in some normed space, convergence of the resulting series results in a well-defined function on that space.

Relations between formal power series are more compactly described in terms of the term-extracting operation $\ext{\ }$ defined by
\begin{equation}\label{eq:bell.4}
\ext{X^m}\sum_{n\ge 0} a_n\,X^n\;=\; a_m\;.
\end{equation}
The following are operations between formal power series that correspond to well known operations between polynomials or between analytical series.  
\smallskip

\paragraph{\bf Powers:} The $n$-th power of a series is defined by
\begin{equation}\label{eq:bell.4.1}
\ext{X^m}\underline a(X)^n\;=\; \sum_{\begin{array}{c}\scriptstyle(k_1,\ldots,k_n)\\ \scriptstyle k_i\ge 0\\ \scriptstyle k_1+\cdots+k_n=m\end{array}} a_{k_1}\cdots a_{k_n}\;.
\end{equation}
\smallskip

\paragraph{\bf Multiplicative inverse:} If the ring $R$ has unit 1, the series $\underline\delta = \delta_{n,0}$ is the unit of the ring $\rser$.  
A formal series $\underline a(X)$ has multiplicative inverse $\underline a^{-1}=:1/\underline a$ if and only if $a_0$ is invertible. The inverse satisfies the recursive equations
\begin{eqnarray}\label{eq:bell.5}
(a^{-1})_0 &=& \frac{1}{a_0}\nonumber\\
(a^{-1})_n &=& -\frac{1}{a_n}\sum_{k=1}^n a_k (a^{-1})_{n-k}\quad,\quad n\ge 1\;.
\end{eqnarray}
This multiplicative inverse, in turns, leads to the natural definition of negative powers of $\underline a(X)$.
\smallskip

\paragraph{\bf Composition:} Operations \eqref{eq:bell.1}-\eqref{eq:bell.2} can be used to develop the series obtained by replacing $X$ by another series. For this operation to make sense, however,  the coefficients of the composed series must involve only a finite number of coefficients of the series being composed.  This happens if the internal series has no constant term.  The  composition of two series $\underline a$ and $\underline b$, with $b_0=0$ is, thus, formally defined by 
\begin{equation}\label{eq:bell.6}
\ext{X^m}\underline a\bigl(\underline b(X)\bigr)\;=\; \sum_{k= 1}^m a_k \sum_{\begin{array}{c}\scriptstyle j_1\ge 1,\ldots,\,j_k\ge 1\\ \scriptstyle j_1+\cdots+j_k=m\end{array}} b_{j_1}\cdots b_{j_k}\;.
\end{equation}
for $m\ge 1$, and constant term $a_0$.  Notice that the restriction $j_1+\cdots+j_k=m$ implies $j_k\le m-k+1$ By grouping terms differing only on the order of factors we arrive to the alternative expression
\begin{equation}\label{eq:bell.7}
\ext{X^m}\bigl(\underline a\circ \underline b\bigr)(X)\;=\; \sum_{k= 1}^m a_k \hspace{-.5cm} \sum_{\begin{array}{c}\scriptstyle \alpha_1\ge 1,\ldots,\,\alpha_{m-k+1}\ge 1\\ \scriptstyle \alpha_1+\cdots+\alpha_{m-k+1}=k \\ \scriptstyle 1\alpha_1+2\alpha_2+\cdots + (m-k+1) \alpha_{m-k+1}=m\end{array}} \hspace{-.5cm}\binom{k}{\alpha_1\ \cdots\ \alpha_{m-k+1}}\,b_1^{\alpha_1}\cdots b_{m-k+1}^{\alpha_{m-k+1}}\;.
\end{equation}
which is one of the variants of the Fa\`a di Bruno's formula. The second sum on the left defines, in fact, the \emph{partial ordinary Bell polynomials} $\hat B_{m,k}$ \cite{Com70}, so
\begin{equation}\label{eq:bell.7.1}
\ext{X^m}\bigl(\underline a\circ \underline b\bigr)(X)\;=\; \sum_{k= 1}^m a_k \,\hat B_{m,k}(b_1,\ldots, b_{m-k+1})\;.
\end{equation}
\smallskip

\paragraph{\bf Formal differentiation:}  The operation is term-wise defined in the obvious way: 
\begin{equation}\label{eq:bell.7.2}
D\underline a(X)\;=\;\underline a'(X)\;:=\; \sum_{n\ge 0} n\, a_n\, X^{n-1}\;.
\end{equation}
As the normal derivative, the operation is linear, satisfies Leibnitz product rule and the chain rule: $(\underline a\circ\underline b)'=[\underline a'\circ b]\underline b'$.

\subsection{Formal Laurent series. Lagrange inversion formula}
The previous construction can be extended to series including a finite number of negative powers of $X$, namely to sequences $\underline a=(a_n)_{n\in\mathbb{Z}}$ such that exists a finite integer $\mathrm{ord}(\underline a)=\min\{n: a_n\neq 0\}$.  These series, together with the operations \eqref{eq:bell.1.1}--\eqref{eq:bell.2.1} forms the ring of \emph{formal Laurent series}, denoted $\rlaur$.  All the previous operations extend naturally to this larger ring.  Notice that every Laurent series is of the form
\begin{equation}\label{eq:bell.7.3}
\underline a(X)\;=\; X^m\,\utilde a(X) \quad,\quad \utilde a(X)\in \rser \quad,\quad \utilde a_0\neq 0
\end{equation}
with $m=\mathrm{ord}(\underline a)$.  Among the terms of Laurent series, the one corresponding to $X^{-1}$ plays a special role.  In analogy to the theory of  analytic functions, the coefficient of this term is called the \emph{residue}, so that $\res(\underline a)=a_{-1}$.  Here are some useful remarks:
\begin{itemize}
\item[(P1)] Every Laurent series takes the form $\underline a= a_{-1} X^{-1} + \underline A$ with $\res(\underline A)=0$.
\item[(P2)]  $\res(\underline a')=0$ for any Laurent series $\underline a$.
\item[(P3)]  $\res(\underline a\,\underline b') = - \res(\underline a'\,\underline b)$.  This is a consequence of (P2) plus the fact that $\underline a\,\underline b'+\underline a'\,\underline b= (\underline a\,\underline b)'$.
\item[(P4)]  Residues allow to turn extraction into product:
\begin{equation}\label{eq:bell.8}
[X^k]\,\underline a(X)\;=\; \res\bigl(X^{k-1} \utilde a\bigr)\;.
\end{equation}
\end{itemize} 
For our purposes, however, the main uses of Laurent series ---and the handling of residues--- is the determination of \emph{compositional inverses} of formal power series.  The issue is to find, for a given series $\underline b\in\rser$ with $b_0=0$, another series $\sdag{b}$ such that $(\sdag{b}\circ\underline b)(X)= X$.  A quick check of the composition algorithm \eqref{eq:bell.6} shows that $\sdag b$ exists if and only if $\underline b$ is of the form
\begin{equation}\label{eq:bell.9}
\underline b(X)\;=\; X\,\utilde b(X)\quad , \quad \utilde b(X)\in \rser\quad, \quad b_0\neq 0\;,
\end{equation}
\emph{provided that} all the coefficients $b_i$ be invertible.  For this reason the theory of compositional inverses is done assuming that $R$ is in fact a field $K$ without .  In this case, every formal Laurent series $\underline a\in \klaur$ has a multiplicative inverse $\underline a^{-1}\in \klaur$, so $\klaur$ becomes a field with the operations \eqref{eq:bell.1.1}--\eqref{eq:bell.2.1}. In this framework we have the following useful property:
\begin{itemize}
\item[(P5)] For any $\underline a\in \klaur$
\begin{equation}\label{eq:bell.10}
\res\bigl(\underline a'/\underline a\bigr)\;=\; \mathrm{ord}(\underline a)\;.
\end{equation}
\end{itemize}
Indeed, if $m=\mathrm{ord}(\underline a)$, then by \eqref{eq:bell.7} and Leibnitz rule
\begin{equation}\label{eq:bell.11.-1}
\underline a'/\underline a\;=\; \frac{mX^{m-1}\,\utilde a + X^m\,\utilde a'}{X^m\,\utilde a}
\;=\; \frac{m}{X} + \bigl(\utilde a'/\utilde a \bigr)\;.
\end{equation}
The last term has no residue because both $\utilde a'$ and $\utilde a^{-1}$ belong to $\kser$.
 \smallskip
 
The last properties needed to determine the compositional inverse requires that, in addition, $K$ be $\mathbb{Q}$-linear.  Then all terms in a Laurent series, except the residual one, can be written as derivatives: $X^k=\bigl[( X^{k+1}/(k+1)\bigr]'$. Hence, property (P1) becomes
\begin{itemize}
\item[(P$1'$)] Every Laurent series takes the form 
\begin{equation}\label{eq:bell.11}
\underline a= a_{-1} X^{-1} + \underline A'
\end{equation}
for some $\underline A\in \klaur$.
\end{itemize}
Note that, as a consequence, property (P2) becomes
\begin{itemize}
\item[(P$2'$)]  $\res(\underline a)=0$ if and only if $\underline a=\underline A'$ for some $\underline A\in \klaur$.
\end{itemize}
The last preliminary property is
\begin{itemize}
\item[(P6)] For any $\underline a\in \klaur$
\begin{equation}\label{eq:bell.12}
\res\bigl[\bigl(\underline a\circ \underline b\bigr)\underline b'\bigr]\;=\; \mathrm{ord}(\underline b)\,\res(\underline a)\;.
\end{equation}
\end{itemize}
Indeed, by property (P$1'$),
\begin{equation}\label{eq:bell.13}
\bigl[\bigl(\underline a\circ \underline b\bigr)\underline b'\bigr]\;=\; a_{-1} \bigl(\underline b'/\underline b\bigr) + \bigl(\underline A\circ \underline b\bigr)'
\end{equation}
and \eqref{eq:bell.12} follows from (P5) and (P1).
\smallskip

The preceding properties are put to good use in the proof of the main result of this section:
\begin{theorem}[\bf Lagrange inversion formula]
Assume $K$ is a $\mathbb Q$-linear field (e.g.\ $\mathbb Q$, $\mathbb R$ or $\mathbb C$).
Let $\underline b, \underline c \in \kser$, with $\underline b$ of the form \eqref{eq:bell.9}.  Then, for all $k\in \mathbb{N}_+$, 
\begin{equation}\label{eq:bell.14}
k[X^k]\,\underline c\circ \sdag b \;=\; [X^{k-1}]\,\underline c'\,\utilde b^{-k}\;.
\end{equation}
In particular, the coefficients of the compositional inverse satisfy
\begin{equation}\label{eq:bell.15}
[X^k]\,\sdag b \;=\; \frac 1k\,[X^{k-1}]\,\utilde b^{-k}\;.
\end{equation}
\end{theorem}
\begin{proof}
We apply (P4) and (P6) to obtain:
\begin{eqnarray}\label{eq:bell.16}
k[X^k]\,\underline c\circ \sdag b &=& k\,\res\bigl[X^{-k-1}\,\underline c\circ \sdag b\bigr]
\;=\; k\,\res\Bigl[\bigl[(X^{-k-1}\,\underline c\circ \sdag b)\circ \underline b\bigr] \,\underline b'\Bigr]\nonumber\\
&=& \res\bigl[k\,\underline b^{-k-1}\,\underline b'\,\underline c\bigr]
\;=\; -\res\bigl[ (\underline b^{-k-1})'\underline c\bigr]\;.
\end{eqnarray}
To conclude we apply (P3) and (P4):
\begin{equation}\label{eq:bell.17}
 -\res\bigl[ (\underline b^{-k-1})'\underline c\bigr]\;=\; \res\bigl[ \underline b^{-k-1}\underline c'\bigr]
\;=\; [X^{k-1}]\,\underline c'\,\utilde b^{-k}\;.
\end{equation}
\end{proof}

\subsection{Virial coefficients and Bell polynomials} Both the series \eqref{eq:rr.fce1} and \eqref{eq:virial3} correspond to formal series $\uhat a$ of the form 
\begin{equation}\label{eq:bell.18}
 \widehat a_n\;=\; \frac{a_n}{n!}\;.
\end{equation}
Such series are often referred to as \emph{exponential power series}.
Notice that, as formal differentiation acts on formal series analogously to regular derivatives on analytic functions,
\begin{equation}\label{eq:bell.19}
  a_n\;=\; D^n\, \uhat a(X)\bigr|_{X=0}
\end{equation}
and, furthermore,
\begin{equation}\label{eq:bell.20}
  (D\underline a)_n\;=\; a_{n+1}\;.
\end{equation}
Formula \eqref{eq:bell.7} implies that the composition $\uhat c=\uhat a \circ \uhat b$ of two such series takes the form
\begin{eqnarray}\label{eq:bell.20.ex}
c_n&=& \sum_{k= 1}^n a_k \hspace{-.5cm} \sum_{\begin{array}{c}\scriptstyle \alpha_1\ge 1,\ldots,\,\alpha_{n-k+1}\ge 1\\ \scriptstyle \alpha_1+\cdots+\alpha_{n-k+1}=k \\ \scriptstyle 1\alpha_1+2\alpha_2+\cdots + (n-k+1) \alpha_{n-k+1}=n\end{array}}  \hspace{-.5cm}\frac{n!}{\prod_{i=1}^{n-k+1} \alpha_i!}\,\prod_{i=1}^{n-k+1} \left[\frac{b_i}{i!}\right]^{\alpha_i}\nonumber\\[10 pt]
&=:& \sum_{k= 1}^n a_k \, B_{n,k}(b_1,\ldots, b_{n-k+1})
\end{eqnarray}
where the factors $ B_{n,k}$ are the \emph{partial exponential Bell polynomials}.

The use of Lagrange inversion formula requires the coefficients of negative powers of a series $\underline b$ with $ b_0=1$. This corresponds to the composition with the series 
\begin{equation}\label{eq:bell.21}
(1+X)^r\;=\; \sum_{n\ge 0}\binom{r}{n}\,X^n
\end{equation}
valid for any real $r$, in which $\binom{r}{n}$ is the generalized binomial coefficient
\begin{equation}\label{eq:bell.22}
\binom{r}{n}\;=\; \frac{r(r-1)\cdots (r-n+1)}{n!}  
\end{equation}
for $n\ge 1$ and $\binom{r}{0}=1$ for all $r\in\mathbb{R}$.  
The series \eqref{eq:bell.22} is readily obtained using the algorithm \eqref{eq:bell.19}.  Composing \eqref{eq:bell.21} with the series $\underline b -1$ and applying \eqref{eq:bell.20.ex} we obtain that if $b_0=1$ and $r\in\mathbb R$,
\begin{equation}\label{eq:bell.23}
\Bigl(\sum_{n\ge 0} \frac{b_n}{n!} \,X^n\Bigr)^r\;=\; \sum_{n\ge 0} \frac{c_n}{n!} \,X^n
\end{equation}
with $c_0=1$ and
\begin{equation}\label{eq:bell.24}
c_n\;=\; \sum_{k= 1}^n \binom{r}{k} \,k!\, B_{n,k}(b_1,\ldots, b_{n-k+1})\quad,\ n\ge 1\;.
\end{equation}
Finally we can prove the needed expression for the virial coeeficients.  

\begin{proposition}The coefficients $\beta_n$ of the virial expansion \eqref{eq:rr.ve} is obtained from the coefficients $b_n$ of the cluster expansion \eqref{eq:rr.ce} through the identity
\begin{equation}\label{eq:rel4}
\beta_{n+1}\;=\;\sum_{k=1}^n\binom{-n}{k}\,k!\,  B_{n,k}(b_2,\ldots,b_{n-k+2})\;.
\end{equation}
\end{proposition}
\begin{proof} 
Denote $\uhat{\beta}$ and $\uhat b$ respectively the formal virial and cluster series, and $\uhat c$ the density series defined by \eqref{eq:rr.fce1}.  They are related by the identities
 \begin{equation}\label{eq:bell.25}
\uhat{\beta}\;=\; \uhat b \circ \uhat c\qquad,\quad \uhat c(X)=X\uhat b'(X)
\end{equation}
Hence, by the Lagrange inversion formula,
 \begin{equation}\label{eq:bell.26}
[X^{n+1}]\,\uhat{\beta}\;=\; \frac{1}{n+1} [X^n] \,\uhat b' (\uhat b')^{-(n+1)}\;=\; \frac{1}{n+1} [X^n] (\uhat b')^{-n}\;.
\end{equation}
Identity \eqref{eq:rel4} follows from \eqref{eq:bell.24} and \eqref{eq:bell.20}.
\end{proof}

\section{Proof of the virial expansion results}\label{sec:pf.virial}
In this section we turn to the proof of the results of Section  \ref{s:rr3}. We need only to prove the Theorem \ref{th:rr1}.  The treatment works for general measure spaces $(S,dx)$ and possibly complex two-body weights with $\card{w(x,y)}\le 1$.

\subsection{Virial coefficients and families of trees}\label{ssec:ct}
Let us start with a well known combinatoric interpretation of Bell polynomials.
\begin{proposition}\label{prop:bell}
 \begin{equation}\label{eq:bell.27}
B_{n,k}(c_1,\ldots,c_{n-k+1})\;= \sum_{\{I_1,\ldots,I_k\} \atop \mathrm{n.t.\,part.\;of}\,\{1,\ldots, n\}}
c_{\card{I_1}} \cdots c_{\card{I_k}}\;.
\end{equation}
\end{proposition}
\begin{proof} It amounts to the observation that each factor
\begin{equation}\label{eq:bell.28}
\frac{n!}{\prod_{i=1}^{n-k+1} [\alpha_i!\,(i!)^{\alpha_i}]}
\end{equation}
is the number of partitions of $\{1,\ldots, n\}$ into $k$ sets, $\alpha_i$ of which have size $i$.  Thus, the left-hand side in \eqref{eq:bell.28} corresponds to decomposing the sum of the right-hand side according of how many sets of each size are present in the partition.
\end{proof}

From \eqref{eq:rr31.0} (with $\Lambda=S$) and Proposition \ref{prop:rr100.part} we see that
\begin{equation}\label{eq:bell.29}
b_n\;=\; \sum_{\tau\in\trees^\star[n-1]}  W_R(\tau)\;.
\end{equation}
Here $\trees^\star[n-1]$ is the family of trees rooted in $\star$ with additional vertices $\{1,\ldots,n-1\}$, $R$ a partition scheme and
\begin{equation}\label{eq:bell.30}
 W_R(\tau)\;=\;\int_{S^{n-1}} 
\prod_{\{i,j\}\in E(\tau)}\bigl[w(x_i,x_j)-1\bigr]\prod_{\{i,j\}\in E(R(\tau))\setminus E(\tau)}w(x_i,x_j)
\, dx_1\cdots dx_{n-1}\;.
\end{equation}
Due to the assumed translation invariance the expression is independent of the choice of $\star$.  We shall henceforth choose $\star=0$ for expressions involving only weights $W_R(\tau)$ without any further detail on the tree.
Combining \eqref{eq:rel4}, \eqref{eq:bell.27} and \eqref{eq:bell.29} we obtain the following expression for the virial coefficients.

\begin{lemma}\label{lemma8} For any $n\ge 1$, 
\begin{equation}\label{eq:rel8}
\beta_{n+1}\;=\; \sum_{k=1}^n\binom{-n}{k}
\sum_{\begin{array}{c}\scriptstyle (I_1,\ldots,I_k)
\\ \scriptstyle  \mathrm{n.t.\,part.\;of}\,\{1,\ldots, n\}\end{array}}
\sum_{\begin{array}{c}\scriptstyle\tau_1\in\trees^{0}[I_1]\\ \scriptscriptstyle  \vdots\\ \scriptstyle\tau_k\in\trees^{0}[I_k]\end{array}} W_{R}(\tau_1)\cdots W_{R}(\tau_k)\;
\end{equation}
Here $\trees^{0}[I_j]$ is the set of trees rooted in $0$ with additional vertices in $I_j$. Also, ``$(I_1,\ldots,I_k)$  $\mathrm{n.t.\,part.\;of}$ $\{1,\ldots, n\}$" means that, in fact, $\{I_1,\ldots,I_k\}$ is a partition of $\{1,\ldots,n\}$ with $I_i\neq\emptyset$.
\end{lemma}
[Note that the factor $k!$ in \eqref{eq:rel4} is used to pass from a sum over families $ \{I_1,\ldots,I_k\}$ to a sum over $k$-tuples $(I_1,\ldots,I_k)$.]

\subsection{Concatenations and splittings of trees}\label{sec:unspit}
The operation of concatenation of trees allows to write \eqref{eq:rel8} as a sum over a single rooted tree with $n$ non-root vertices.  The operation, first applied in \cite{ram15,ramtat15} where it was called \emph{merging}, exploits the fact that the weight $W_R(\tau)$ are invariant to the choice of the root.  Hence roots can be concatenated so to place the trees $\tau_1,\ldots,\tau_k$ one after the other in order to form a single tree.  This can be done in many different ways, but we are interested in obtaining a tree $\tau$ whose weight $W(\tau)$ is the product of the individual trees.  

\begin{definition} Let $\{I_1,\ldots,I_k\}$ be a partition $\{1,\ldots,n\}$, $\star_1,\ldots,\star_k$ points in $S$ and $R$ a partition scheme.  
\begin{itemize}
\item A \emph{concatenation} of $\star$-rooted trees with non-root vertices in the partition is a map
\begin{equation}\label{eq:bell.31}
\mathrm{M}:\prod_{i=1}^{k}\mathcal{T}^{\star_i}[I_i]\longrightarrow\mathcal{T}^0[n]\qquad,\quad M(\tau_1,\ldots,\tau_k)= \tau\;.
\end{equation}
Note that the order in which the trees are concatenated may matter.

\item The concatenation is $R$-\emph{faithful} if 
\begin{equation}\label{eq:bell.32}
W_R(\tau)\;=\; W_R(\tau_1)\cdots W_R(\tau_k)\;.
\end{equation}
\end{itemize}
\end{definition}

\begin{definition}[\bf Penrose concatenation] \cite{ram15,ramtat15} Let $\{I_1,\ldots,I_k\}$ be a partition $\{1,\ldots,n\}$ and, for $\tau$ a rooted tree, let us denote $j_{\max}(\tau)$ the largest label of the vertices at  maximal (tree) distance from the root.   The \emph{Penrose concatenation} is defined by the map 
\begin{equation}\label{eq:bell.33}
\mathrm{PFM}(\tau_1,\ldots,\tau_k)\;=\; \bigcup_{i=1}^k \tau'_i
\end{equation}
where $\tau'_1,\ldots,\tau'_k$ are the trees obtained from  $\tau_1,\ldots,\tau_k$ by choosing $\star_1=0$ (a fix point in $S$) and $\star_i=j_{\max}(\tau_{i-1})$ for $j=2,\ldots,k$.  [The union of graphs is the graph formed by taking the union of the vertices and the union of the links.]
\end{definition}
Figure \ref{fig:merging} shows a simple example of this operation.  It is clear that the operation indeed yields a tree, but it is far from one-to-one.  

\begin{figure}
  \includegraphics[width=\linewidth]{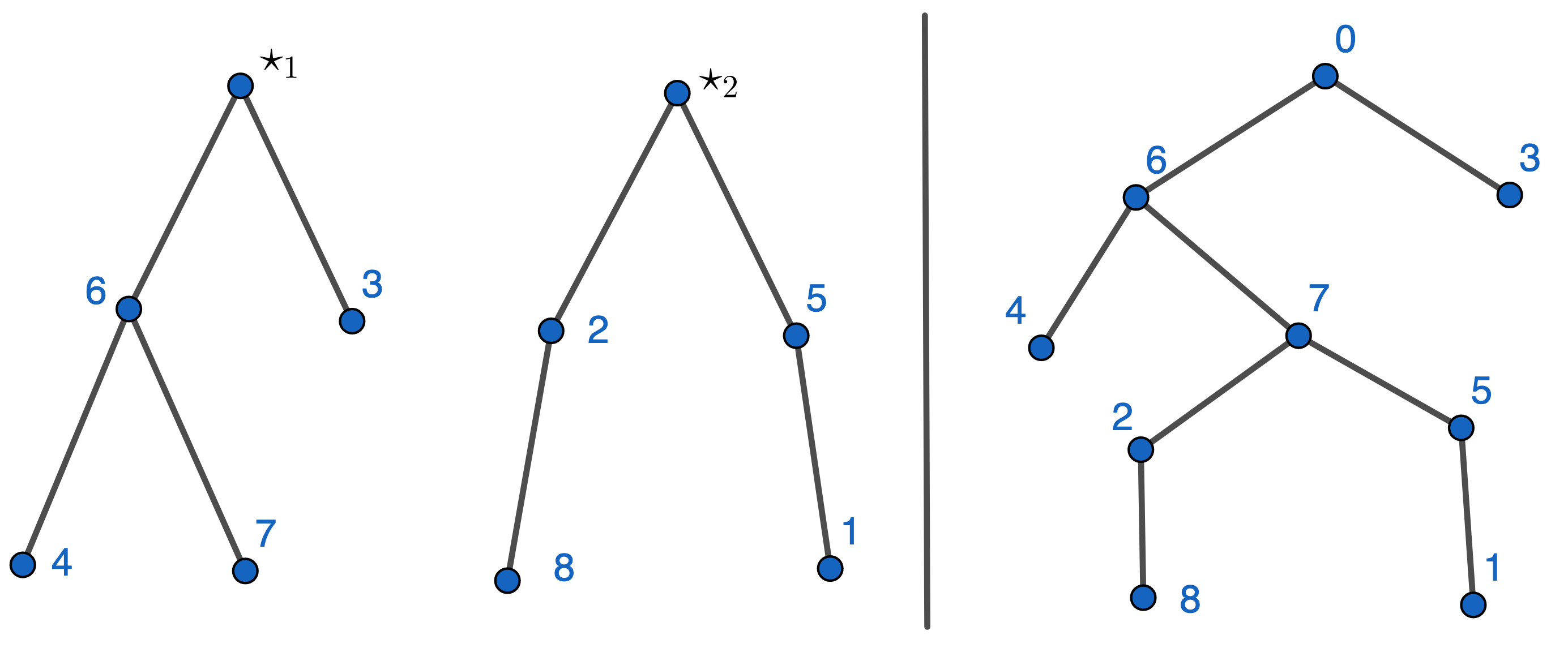}
  \caption{An example of Penrose concatenation. The labels
in the vertices on the left-hand side are the original labels. On the
right-hand side we have the concatenation of these trees.}
  \label{fig:merging}
\end{figure}

\begin{lemma} Consider the Penrose concatenation \eqref{eq:bell.33}. 
\begin{itemize}
\item[(I)] The following properties hold:
\begin{itemize}
\item[(i)] $E(\tau)=\cup_{i=1}^kE(\tau_i')$,
\item[(ii)] $[E\big(R_{\mathrm{Pen}}(\tau)\big)\big]\setminus E(\tau)=\cup_{i=1}^k[E\big(R_{\mathrm{Pen}}(\tau_i')\big)\big]\setminus E(\tau_i')$.
\end{itemize}
\item[(II)] PFM is $R_{\mathrm{Pen}}$-faithful
\end{itemize}
\end{lemma}
\begin{proof}
{(I)} Property (i) is immediate from the definition.  It is enough to prove property (ii) for $k=2$.  In this case, the only possible difference between  $[E\big(R_{\mathrm{Pen}}(\tau)\big)\big]\setminus E(\tau)$ and $\cup_{i=1}^2[E\big(R_{\mathrm{Pen}}(\tau_i')\big)\big]\setminus E(\tau_i')$ can come from links connecting the children of the root of $\tau'_2$ and the remaining leaves of $\tau'_1$.  These link, would belong to $E\big(R_{\mathrm{Pen}}(\tau)\big)$ only if there were a leaf of $\tau'_1$ with vertex larger than $j_{\max}$.  This does not happen precisely because $j_{\max}$ was chosen as the leaf of $\tau'_1$ with the \emph{largest} vertex.
\smallskip

{(II)} Again, it is enough to consider $k=2$.  By translation invariance $W_{R_{\mathrm{Pen}}}(\tau_i)=W_{R_{\mathrm{Pen}}}(\tau_i')$.  The identity  
\begin{equation}\label{eq:mer1}
W_{R_{\mathrm{Pen}}}(\tau)\;=\; W_{R_{\mathrm{Pen}}}(\tau_1')\,W_{R_{\mathrm{Pen}}}(\tau_2')
\end{equation}
is obtained by integrating first, in the left-hand side, over the variables $x_i:i\in I_2$.  By translation invariance this integral results in a factor $W_{R_{\mathrm{Pen}}}(\tau_2')$ and the remaining integration produces the factor $W_{R_{\mathrm{Pen}}}(\tau_1')$.
\end{proof}
\smallskip

From now on we return to the choice $\star_i=0$ for all $i$. This is no loss of generality due to  translation invariance.

\begin{definition} A \emph{$k$-splitting} of a tree 
$\tau\in \mathcal{T}^0[n]$ is an ordered family of $k$ trees whose concatenation yields   $\tau$, namely a family 
\begin{equation}\label{eq:bell.34}
(\tau_1,\ldots,\tau_k)\in\prod_{i=1}^k\mathcal{T}^{0}[I_i]:
 \{I_1,\ldots,I_k\} \, \mathrm{n.t.\,part.\,of}\,\{1,\ldots, n\} \,\mathrm{and}\,
\mathrm{PFM}(\tau_1,\ldots,\tau_k)\;=\;\tau\;.
\end{equation}
The set of $k$-splittings of $\tau$ will be denoted $\mathrm{Sp}_k(\tau)$.
\end{definition}

The faithfulness of the Penrose concatenation allows the rewriting of \eqref{eq:rel8} in the form
\begin{equation}\label{eq:bell.35}
\beta_{n+1}\;=\; \sum_{k=1}^n\binom{-n}{k}\sum_{\tau\in\mathcal{T}^0[n]} \card{\mathrm{Sp}_k(\tau)} \,W_{R_{\mathrm{Pen}}}(\tau)
\end{equation}
for $n\ge 1$.
To control the cardinality of $\mathrm{Sp}_k(\tau)$ we follow \cite{ram15,ramtat15} and each tree $\tau$ according to its maximal number of splittings.  
\begin{definition} A tree $\tau\in \mathcal{T}^0[n]$ is \emph{Penrose $\ell$-splittable} if 
$\mathrm{Sp}_\ell(\tau)\neq\emptyset$ but $\mathrm{Sp}_{\ell+1}(\tau)=\emptyset$.   The set of these trees will be denoted $\mathcal{T}^0_{\mathrm{Pen},\,\ell}[n]$.  The Penrose 1-splittable trees will also be called \emph{unspittable trees}.   
\end{definition}
A $\ell$-splittable tree can be written as the concatenation of a smaller number of trees only if they are obtained by grouping some consecutive components of the maximal concatenation.  This can be done in $\binom{\ell-1}{k-1}$ ways. Therefore \eqref{eq:bell.35} yields
\begin{eqnarray}\label{eq:bell.36}
\beta_{n+1} &=& \sum_{k=1}^n\binom{-n}{k}\sum_{\ell=k}^{n}\sum_{\tau\in\mathcal{T}^0_{\mathrm{Pen},\,\ell}[n]} \binom{\ell-1}{k-1}\,W_{R_{\mathrm{Pen}}}(\tau)\nonumber\\
&=& \sum_{\ell=1}^n\Bigl[\sum_{k=1}^{\ell}\binom{-n}{k}\binom{\ell-1}{k-1}\Bigr] \sum_{\tau\in\mathcal{T}^0_{\mathrm{Pen},\,\ell}[n]}  \,W_{R_{\mathrm{Pen}}}(\tau)\nonumber\\
&=& \sum_{\ell=1}^n\Bigl[\sum_{k=0}^{\ell-1}\binom{-n}{k+1}\binom{\ell-1}{k}\Bigr] \sum_{\tau\in\mathcal{T}^0_{\mathrm{Pen},\,\ell}[n]}  \,W_{R_{\mathrm{Pen}}}(\tau)
\end{eqnarray}

The following lemma leads to the final expression.
\begin{lemma}\label{lemma4} \cite{ram15,ramtat15}
For any $r\in\mathbb{R}$ and any $m \in\mathbb{N}$, 
\begin{equation}\label{eq:bell.37}
\sum_{k=0}^{m}\binom{r}{k+1}\binom{m}{k}\;=\;\binom{m+r}{m+1}\;.
\end{equation}
[Binomials are understood in the general sense \eqref{eq:bell.22}.]
\end{lemma}
\begin{proof}
Writing the left-hand side of \eqref{eq:bell.37} in the form
\begin{equation}\label{eq:bell.38}
\sum_{k=0}^{m}\binom{r}{k+1}\binom{m}{m-k}
\end{equation}
we see that it corresponds to the $m$-th coefficient of the product of two series.  Resorting to 
\eqref{eq:bell.21} we deduce that
\begin{equation}\label{eq:bell.39}
\binom{r}{k+1}\;=\; [X^k]\,\frac1X \bigl[(1+X)^r-1\bigr]
\end{equation}
and
\begin{equation}\label{eq:bell.40}
\binom{m}{m-k}\;=\; [X^{m-k}]\,(1+X)^m\;.
\end{equation}
Hence,
\begin{eqnarray}\label{eq:bell.41}
\sum_{k=0}^{m}\binom{r}{k+1}\binom{m}{k}&=& [X^m]\frac1X \bigl[(1+X)^r-1\bigr](1+X)^m\nonumber\\
&=& [X^{m+1}]\bigl[(1+X)^{r+m}-(1+X)^m\bigr]\\
&=& \binom{m+r}{m+1} - 0\;.\nonumber
\end{eqnarray}
The last line is also a consequence of \eqref{eq:bell.21}.
\end{proof}
 
 For $r=-n$ and $m=\ell-1$ the lemma yields
 \begin{equation}\label{eq:bell.42}
\sum_{k=0}^{\ell-1}\binom{-n}{k+1}\binom{\ell-1}{k}\;=\;\binom{-n+l-1}{\ell}\;=\; (-1)^\ell\binom{n}{\ell}
\end{equation}
and \eqref{eq:bell.36} becomes
\begin{equation}\label{eq:bell.43}
\beta_{n+1}\;=\; \sum_{\ell=1}^n (-1)^\ell\binom{n}{\ell}
\sum_{\tau\in\mathcal{T}^0_{\mathrm{Pen},\,\ell}[n]}  \,W_{R_{\mathrm{Pen}}}(\tau)
\end{equation}
for any $n\ge 1$ [$\beta_1=1$].

\subsection{Proof of Theorem \ref{th:rr1}}

We introduce the formal series 
\begin{equation}\label{eq:thevir3}
T^{\mathrm{Pen},\ell}_{W}(X)\;=\;\sum_{n=1}^{\infty}\frac{X^n}{n!}\sum_{\tau\in\mathcal{T}^0_{\mathrm{Pen},\ell}[n]} W_{R_{\mathrm{Pen}}}(\tau)
\end{equation}
and rewrite \eqref{eq:bell.43} in the form
\begin{equation}\label{eq:bell.43.1}
\beta_{n+1}\;=\; [X^n] \sum_{\ell=1}^n (-1)^\ell\binom{n}{\ell} T^{\mathrm{Pen},\ell}_{W}(X)\;.
\end{equation}

\begin{lemma}\label{lemma6} For $m\in\mathbb{N}$, we have
\begin{equation}\label{eq:mer10}
T^{\mathrm{Pen},\ell}_W(X)\;=\;\left[T^{\mathrm{Pen},1}_W(X)\right]^\ell.
\end{equation}
\end{lemma}
\begin{proof}
Each $\tau\in\mathcal{T}^0_{\mathrm{Pen},\ell}$ admits exactly one decomposition $(\tau_1,\ldots,\tau_\ell)$ with $\tau_i\in \tau\in\mathcal{T}^0_{\mathrm{Pen},1}[I_i]$ for some partition $(I_1,\ldots,I_k)$.  Hence,
\begin{equation}\label{eq:bell.44}
\sum_{\tau\in\mathcal{T}^0_{\mathrm{Pen},\ell}[n]} W_{R_{\mathrm{Pen}}}(\tau)\;=\;
\sum_{\begin{array}{c}\scriptstyle (I_1,\ldots,I_k)
\\ \scriptstyle  \mathrm{n.t.\,part.\;of}\,\{1,\ldots, n\}\end{array}}
\sum_{\tau_1\in\mathcal{T}^0_{\mathrm{Pen},1}[I_1]} W_{R_{\mathrm{Pen}}}(\tau_1)
\cdots
\sum_{\tau_\ell\in\mathcal{T}^0_{\mathrm{Pen},1}[I_\ell]} W_{R_{\mathrm{Pen}}}(\tau_\ell)
\end{equation}
Each sum 
\[
t_{\card I}^{(m)}\;:=\:\sum_{\tau\in\mathcal{T}^0_{\mathrm{Pen},m}[I]} W_{R_{\mathrm{Pen}}}(\tau)
\] 
depends only on the cardinality of the index set $ I$.  Hence, \eqref{eq:bell.44} becomes
\begin{equation}\label{eq:bell.45}
t_{n}^{(\ell)}\;=\; \sum_{\begin{array}{c}\scriptstyle(m_1,\ldots,m_\ell)\\ \scriptstyle m_1+\cdots+m_n=\ell\end{array}} \binom{n}{m_1\,\ldots\,m_\ell}\, t_{m_1}^{(1)} \cdots t_{m_\ell}^{(1)}
\end{equation}
which, by \eqref{eq:bell.4.1}, proves \eqref{eq:mer10}.
\end{proof}

\begin{lemma}\label{lemma3}
The following identity between formal power series holds
\begin{equation}\label{eq:mer21}
T^{\mathrm{Pen},1}_W(X)\;=\;1-\Beta_W^{-1}(X)\;.
\end{equation}
where the coefficients of $\Beta_W(X)$ are defined in \eqref{eq:rr41.1.gg}. 
\end{lemma}

\begin{proof}
The identity
\begin{equation}\label{eq:mer22}
\sum_{m=1}^n\sum_{\tau\in\mathcal{T}^0_{\mathrm{Pen},m}[n]} W(\tau)\;=\;\sum_{\tau\in\mathcal{T}^0[n]} W(\tau)
\end{equation}
implies 
\begin{equation}\label{eq:mer25}
n!\,[X^n]\sum_{m=1}^nT^{\mathrm{Pen},m}_W(X)=n![X^n](\Beta_W(X)-1)\qquad,\quad n\ge 0\;.
\end{equation}
By Lemma \ref{lemma6} this implies
\begin{equation}\label{eq:mer26}
[X^n]\sum_{m=1}^n[T^{\mathrm{Pen},1}_W(X)]^m=[X^n](\Beta_W(X)-1),\qquad n\ge 0\;.
\end{equation}
Formula \eqref{eq:bell.5} readily implies that 
\begin{equation}\label{eq:mer26.1}
[X^n]\sum_{m=0}^n[T^{\mathrm{Pen},1}_W(X)]^m=[X^n]\bigl[1-T^{\mathrm{Pen},1}_W(X)\bigr]^{-1}\;,
\end{equation}
hence \eqref{eq:mer26} yields the identity
\begin{equation}\label{eq:mer26.2}
T^{\mathrm{Pen},1}_W(X) \bigl[1-T^{\mathrm{Pen},1}_W(X)\bigr]^{-1}\;=\; \Beta_W(X)-1
\end{equation}
which can be manipulated into the form  \eqref{eq:mer21}.
\end{proof}

Combining \eqref{eq:bell.43.1} with the last two lemmas we obtain useful expressions for the virial coefficients.
\begin{proposition}\label{prop:rr.10}
For $n\ge 0$, 
\begin{eqnarray}\label{eq:mer30}
\frac{\beta_{n+1}}{(n+1)!} &=& \frac{1}{n+1} \,[X^n] \bigl[1-T^{\mathrm{Pen},1}_W(X)\bigr]^n\\
\label{eq:mer31}
&=& \frac{1}{n+1} \,[X^n]\, \Beta^{-n}(X)\;.
\end{eqnarray}
\end{proposition}
Notice that by \eqref{eq:bell.8} 
\begin{eqnarray}\label{eq:mer30.1}
\frac{\beta_{n+1}}{(n+1)!} &=& \frac{1}{n+1} \,\mathrm{Res} \Bigl[X^{-n-1}\,\bigl[1-T^{\mathrm{Pen},1}_W(X)\bigr]^n\Bigr]\\
\label{eq:mer31.1}
&=& \frac{1}{n+1} \,\mathrm{Res} \Bigl[X^{-n-1}\, \Beta^{-n}(X)\Bigr]\;.
\end{eqnarray}
\smallskip

\subsection*{Conclusion of the roof of Theorem \ref{th:rr1}}
The remaining step is to bound $\card{\beta_{n+1}}$.  This is done through the inequality
\begin{equation}\label{eq:mer35}
[X^n]\,\bigl|1-T^{\mathrm{Pen},1}_W(X)\bigr|\;\le\; [X^n]\;\bigl[1+ T^{\mathrm{Pen},1}(X)\bigl]
\end{equation}
where $T^{\mathrm{Pen},1}(X)$ is defined exactly as $T^{\mathrm{Pen},1}_W(X)$ but using the weights \eqref{eq:rn5.gen}.  Note that, 
\begin{equation}\label{eq:mer36}
T^{\mathrm{Pen},1}(X)\;=\; 1-\Beta^{-1}(X)
\end{equation}
with the coefficients of $\Beta(X)$ defined with the weights \eqref{eq:rn5.gen}.  This identity is proven following the steps of the proof of Lemma \ref{lemma3}.  Finally, for $r$ within the radius of convergence $r^{**}$ of the series $\Beta(X)$, which coincides with that of $T^{\mathrm{Pen},1}(X)$, we can apply Cauchy bound $\card{\mathrm{Res}[f(z)]}\;\le\; r \card{f(r)}$ to \eqref{eq:mer30.1}--\eqref{eq:mer31.1}.  We conclude:
\begin{eqnarray}\label{eq:rrmer40}
\frac{\card{\beta_{n+1}}}{(n+1)!} &\le& \frac{1}{n+1} \biggl[ \inf_{0\le r\le r^{**}}\Bigl(\frac 2r-\frac{1}{r\,\Beta(r)}\Bigr)\biggr]^n\\
\label{eq:rrmer40.1}
&=& \frac{1}{n+1} \biggl[ \inf_{0\le r\le r^{**}}\Bigl(\frac{1+ T^{\mathrm{Pen},1}(r)}{r}\Bigr)\biggr]^n\;.
\end{eqnarray}

The proof of Theorem  \ref{th:rr1} is completed.

\paragraph{\textbf{Acknowledgments}} T. X. Nguyen has been partially supported by the grant of GSSI (Gran Sasso Science Institute). T. X Nguyen and R. Fern\'andez also would like to thank NYU Shanghai (New York University Shanghai)  for supporting the visiting of T. X Nguyen to Shanghai. We also thank  S. Ramawadh and S. J. Tate for sharing with us their novel handling of the virial-expansion coefficients.

\end{document}